\documentclass[]{article}
\usepackage{amsfonts}
\usepackage{amssymb}
\usepackage[letterpaper, left=2.5cm, right=2.5cm, top=2.5cm,
bottom=2.5cm,dvips]{geometry}
\usepackage{verbatim}

\usepackage[T1]{fontenc}
\usepackage{graphicx}
\usepackage{diagbox}
\usepackage{xcolor}
\usepackage{amsmath}
\usepackage{array,multirow}
\usepackage{dsfont}
\usepackage{amsthm}
\usepackage{thmtools}
\usepackage{thm-restate}
\usepackage{enumerate}
\usepackage{bm}

\usepackage[pagebackref]{hyperref}
\hypersetup{
	colorlinks,
	linkcolor={blue!60!black},
	citecolor={green!60!black},
	urlcolor={green!60!black}
}

\usepackage{xcolor}
\definecolor{cream}{RGB}{203, 237, 204}
\newtheorem{theorem}{Theorem}[section]
\newtheorem{lemma}[theorem]{Lemma}

\newtheorem{corollary}[theorem]{Corollary}
\newtheorem{conjecture}[theorem]{Conjecture}
\newtheorem{definition}[theorem]{Definition}
\newtheorem{construction}[theorem]{Construction}
\newtheorem{remark}[theorem]{Remark}
\newtheorem{claim}[theorem]{Claim}

\numberwithin{equation}{section}

\newcommand{\PP}{{\mathcal{P}}}
\newcommand{\A}{{\mathcal{A}}}
\newcommand{\B}{{\mathcal{B}}}
\newcommand{\C}{{\mathcal{C}}}
\newcommand{\G}{{\mathcal{G}}}
\newcommand{\Q}{{\mathcal{Q}}}
\newcommand{\supp}{{\text{supp}}}

\title{On low-power error-correcting cooling codes with large distances}
\author{Yuhao~Zhao\thanks{Y. Zhao ({\tt zhaoyh21@mail.ustc.edu.cn}) is with the School of Mathematical Sciences, University of Science and Technology of China, Hefei, 230026, Anhui, China.} ~and~Xiande~Zhang
\thanks{X. Zhang ({\tt drzhangx@ustc.edu.cn}) is with the School of Mathematical Sciences, University of Science and Technology of China, Hefei, 230026, and with Hefei National Laboratory, University of Science and Technology of China, Hefei 230088, China. The research of X. Zhang is supported by the National Key Research and Development Programs of China 2023YFA1010200 and 2020YFA0713100, the NSFC under Grants No. 12171452 and No. 12231014, and the Innovation Program for Quantum Science and Technology 2021ZD0302902.}  }
\date{}

\begin{document}
	\maketitle
\begin{abstract}
	A low-power error-correcting cooling (LPECC) code was introduced as  a coding scheme for communication over a bus by Chee et al. to control the peak temperature, the average power consumption of on-chip buses, and error-correction for the transmitted information, simultaneously. Specifically, an $(n, t, w, e)$-LPECC code is a coding scheme over $n$ wires that
avoids state transitions on the $t$ hottest wires and allows at most $w$ state transitions in each transmission, and can correct up to $e$ transmission errors.
In this paper, we study the maximum possible size of an $(n, t, w, e)$-LPECC code, denoted by $C(n,t,w,e)$. When $w=e+2$ is large, we establish a general upper bound $C(n,t,w,w-2)\leq \lfloor \binom{n+1}{2}/\binom{w+t}{2}\rfloor$; when $w=e+2=3$, we prove $C(n,t,3,1) \leq \lfloor \frac{n(n+1)}{6(t+1)}\rfloor$. Both bounds are tight for large $n$ satisfying some divisibility conditions. Previously, tight bounds were known only for $w=e+2=3,4$ and $t\leq 2$.
 In general, when $w=e+d$ is large for a constant $d$, we determine the asymptotic value of $C(n,t,w,w-d)\sim \binom{n}{d}/\binom{w+t}{d}$  as $n$ goes to infinity, which can be extended to $q$-ary codes.

% For general $w$ and $e$, we determine the asymptotic values of $C(n,t,w,e)$ under certain conditions and further extend this result to $q$-ary codes. In particular, we have that $C(n,1,w,e)\sim \binom{n+1}{w-e}/\binom{w+1}{w-e}$ for $\frac{\sqrt{5}-1}{2}w \leq e\leq w-1$ as $n$ goes to infinity.

%This paper investigates the maximum possible size of an $(n, t, w, e)$-LPECC code, which is denoted by $C(n,t,w,e)$.
%	{\color{blue}
%	 We first consider the exact values of $C(n,t,w,w-2)$ for various parameters. For $w\geq t^2+2t+2$, we show that $C(n,t,w,w-2)=\binom{n+1}{2}/\binom{w+t}{2}$
%	if $n$ is sufficiently large and satisfies some divisibility conditions. For $w=3$, we prove that $C(n,t,3,1) =\lfloor \frac{n(n+1)}{6(t+1)}\rfloor$ for any positive integers $n,t$ with $n\equiv 2 \pmod{6(t+1)}$ and $n> 2$. Moreover, by using a similar construction, a new lower bound of $C(n,t,4,2)$ is also given. Lastly, we determine the asymptotic values of $C(n,t,w,e)$ under certain conditions and further extend this result to $q$-ary codes. In particular, we have that $\lim_{n\rightarrow \infty}\frac{C(n,1,w,e)}{\binom{n+1}{w-e}/\binom{w+1}{w-e}}=1$ for any fixed parameters $w,e$ with $\frac{\sqrt{5}-1}{2}w \leq e\leq w-1$.
%}

\end{abstract}

\section{Introduction}
Whether targeting battery-powered devices or  high-end systems, power and heat dissipation limits are first-order design constraints for chips.  To reduce the overall power consumption of on-chip buses, there are many encoding techniques that have been proposed, see for example, \cite{benini1997asymptotic,calimera2008thermal,chee2006optimal,petrov2004low,sotiriadis2002bus,sotiriadis2003bus,stan1995bus,wang2007chip} and the references therein.
However, power-aware design alone is not sufficient to address the thermal challenges posed by high-performance interconnects, because it does not directly target the spatial and temporal behavior of the operating environment. Therefore, thermally-aware approaches have become one of the most important domains of research in chip design.

In order to directly control the peak temperature of a bus by cooling its hottest wires, Chee et al.~\cite{chee2017cooling} introduced a new class of codes called {\it cooling code}s. Roughly speaking, in a cooling code, the peak temperature can be controlled by avoiding state transitions on the hottest wires for as long as necessary until their temperatures drop.
Moreover, {\it low-power cooling (LPC) codes} were introduced in \cite{chee2017cooling,chee2020low} to control both the peak temperature and the average power consumption simultaneously. Furthermore, if an LPC code can also control error-correction for the transmitted information, then it is said to be a {\it low-power error-correcting cooling (LPECC) code} \cite{chee2017cooling}.
More precisely, an $(n, t, w, e)$-LPECC code is a coding scheme for communication over a bus consisting of $n$ wires having the following three properties:
\begin{itemize}
	\item Property $\bm A(t)$: every transmission does not cause state transitions on the $t$ hottest wires;
\item Property $\bm B(w)$: the total number of state transitions on all the wires is at most $w$ in every transmission;
\item Property $\bm C(e)$: up to $e$ transmission errors ($0$ received as $1$, or $1$ received as $0$) on the $n$ wires can be corrected.
\end{itemize}

The values of $t, w, e$ are called {\it design parameters} which are determined by the specific thermal requirements of specific interconnects.
Note that codes with Property $\bm C(e)$ are well-studied error-correcting codes. A coding scheme satisfying only Property $\bm A(t)$ is the aforementioned $(n, t)$-cooling code introduced by Chee et al.~\cite{chee2017cooling}. In addition, coding schemes that satisfy Properties $\bm A(t)$ and $\bm B(w)$ simultaneously are called $(n,t,w)$-low power cooling (LPC) codes, which were studied in \cite{chee2017cooling,chee2020low}. Moreover, if a coding scheme with Property $\bm C(e)$ also satisfies Property $\bm A(t)$ or Property $\bm B(w)$, this code is called an $(n, t, e)${\it -error-correcting cooling (ECC) code} or an {\it $(n, w, e)$-low-power error-correcting (LPEC) code}, respectively (see \cite{chee2017cooling}).

In this paper, we are interested in the maximum size of an $(n, t, w, e)$-LPECC code, which is denoted by $C(n,t,w,e)$.
Let us list some known results of $C(n,t,w,e)$ for the reader's convenience.
\begin{itemize}
	\item $n=w(t+1)$: Chee et al.~\cite{chee2017cooling} used resolvable Steiner systems to give constructive lower bounds  on $C(w(t+1),t,w,e)$ with
	(\romannumeral1) $(t,w,e)=(t,3,1)$, where $t$ is even;
	(\romannumeral2) $(t,w,e)=(t,4,1)$, where $t\equiv 0$ or $1 \pmod{3}$;
	(\romannumeral3) $(t,w,e)=(t,4,2)$, where $t\equiv 0 \pmod{3}$;
	(\romannumeral4) $(t,w,e)=(t,w,w-2)$, where $t\equiv 0 \pmod{w-1}$ is sufficiently large;
	(\romannumeral5) $(t,w,e)=(q-1,q,q-2)$, where $q$ is a prime power;	
	   \item $e=w-1$: $C(n, t, w, w-1) = \lfloor \frac{n+1}{w+t} \rfloor$ for any positive integers $t, w \leq n$ \cite[Corollary 4.4]{Liu-Ji};
	   \item $e=w-2$:
 \begin{enumerate}[(a)]
 	  \item $C(n, 1, 3, 1) = \lfloor \frac{n(n+1)}{12}\rfloor$ for any positive integer $n$ with $n \neq 6$, and $C(6, 1, 3, 1) = 2$ \cite[Theorem 4.11]{Liu-Ji};
	   \item $C(n,1,4,2)\leq \lfloor \frac{n(n+1)}{20}\rfloor$ for any positive integer $n\geq 10$ \cite[Theorem 3.3]{Liu-Ji}, and $C(n,1,4,2)= \frac{n(n+1)}{20}$ for any positive integer $n \equiv 0, 4 \pmod{20}$ \cite[Corollary 4.4]{Liu-Ji};
	   \item $C(n,2,3,1)\leq \lfloor \frac{n(n+1)}{18}\rfloor$ for any positive integer $n\geq 9$ \cite[Theorem 3.4]{Liu-Ji},  $C(n, 2, 3, 1) = \lfloor \frac{n(n+1)}{18} \rfloor$ for any positive integer $n \equiv 0 \pmod 2$ with $n \neq 6$, and $C(6, 2, 3, 1) = 1$ \cite[Theorem 4.19]{Liu-Ji}.
	\end{enumerate}
\end{itemize}

%\subsection{Our results}
  We continue the study on the maximum size of an $(n, t, w, e)$-LPECC code with $e=w-2$ and further consider the asymptotic value for more general $w$ and $e$. Our results are summarized below.
  \begin{enumerate}[(i)]
  	\item $C(n,t,w,w-2)\leq \lfloor \binom{n+1}{2}/\binom{w+t}{2}\rfloor$ for any positive integers $n,t,w$ with $w\geq t^2+2t+2$ and $n\geq w+t-1$ (Theorem~\ref{thm-(n,t,w,w-2)-upper}). %For $w\geq t^2+2t+2$, we have $C(n,t,w,w-2)=\binom{n+1}{2}/\binom{w+t}{2}$
  The equality holds	if $n$ is sufficiently large  satisfying $\binom{w+t}{2} \mid \binom{n+1}{2}$ and $w+t-1\mid n$ (Corollary~\ref{cor-value-(n,t,w,w-2)}). Further for such $n$, any $(n,t,w,w-2)$-LPECC code of size meeting the upper bound must come from an $(n+1,w+t,1)$-BIBD (Theorem~\ref{thm-BIBD}).
%
%  Furthermore, we show that any $(n,t,w,w-2)$-LPECC code of size $\binom{n+1}{2}/\binom{w+t}{2}>1$ must come from an $(n+1,w+t,1)$-BIBD, provided $w> t^2+2t+2$ and $\binom{w+t}{2} \mid \binom{n+1}{2}$ (Theorem~\ref{thm-BIBD}).
  %	These extend the aforementioned results for $C(n, 1, 3, 1)$ and $C(n,1,4,2)$ given by Liu and Ji~\cite{Liu-Ji}.
  	
  	\item $C(n,t,3,1) \leq  \lfloor\frac{n(n+1)}{6(t+1)}\rfloor$ for any positive integers $n,t$ with $n\geq 3(t+1)$  (Theorem~\ref{thm-(n,t,3,1)-upper}), and the equality holds
  %$C(n,t,3,1) =\lfloor \frac{n(n+1)}{6(t+1)}\rfloor$ for any positive integers $n,t$ with
   if $n\equiv 2 \pmod{3(t+1)}$ and $2\mid n$ (Corollary~\ref{cor-value-(n,t,3,1)}).  This extends the previous results for $C(n, 1, 3, 1)$ and $C(n,2,3,1)$~\cite{Liu-Ji}. Moreover, we have that $\liminf_{n\to\infty} C(n,t,4,2)/n^2 \geq \frac{1}{12(t+1)}$ (Theorem~\ref{thm-(n,t,4,2)-lower}). These results indicate that the upper bound in (i) is not valid for small $w$ compared with $t$.
   %$C(n,t,w,w-2)\leq \lfloor \binom{n+1}{2}/\binom{w+t}{2}\rfloor$  with $w=4$ is not valid for large $t$.
  	
  	\item  $\lim_{n\rightarrow \infty}\frac{C(n,t,w,e)}{\binom{n+1}{w-e}/\binom{w+t}{w-e}}=1$, where $t,w,e$ satisfies $e\leq w-1$ and $2\binom{w}{w-e}  \geq  \binom{w-t-1}{w-e} + \binom{w+t}{w-e}$ (Corollary~\ref{cor-limit-q-ary-lpecc}). Note that, when $w-e$ is a positive constant, the conditions hold if $w$ is large; when $t=1$, the conditions hold if $\frac{\sqrt{5}-1}{2}w \leq e\leq w-1$ (see Remark~\ref{rmk-condition}). In fact, we show that $\lim_{n\rightarrow \infty}\frac{C_q(n,t,w,e)}{(q-1)^{w-e} \binom{n}{w-e}/\binom{w+t}{w-e}}=1$ under the same condition, where $C_q(n,t,w,e)$ denotes the maximum size of a $q$-ary $(n,t,w,e)$-LPECC code (Corollary~\ref{cor-limit-q-ary-lpecc}).
  Moreover, we discuss bounds and exact values of constant-weight LPECC codes based on known results on constant-weight codes  (Corollary~\ref{cor-value-cpecc}).

%  	  Moreover, we also discuss constant-weight LPECC codes, i.e., constant-power error-correcting cooling (CPECC) codes, which were introduced by Chee et al.~\cite{chee2020low}.  We determine the maximum size of a CPECC code for some parameters if certain generalized Steiner system exists.
%  	In particular, this gives the exact value for binary CPECC codes (Corollary~\ref{cor-value-cpecc}).
  \end{enumerate}

  The rest of this paper is organized as follows. In Section~\ref{sec-preliminary}, we review a combinatorial characterization  and a basic construction of LPECC codes via packings due to Liu and Ji~\cite{Liu-Ji}. In Section~\ref{sec-(n,t,w,w-2)-large-w}, we prove $C(n,t,w,w-2)\leq \lfloor \binom{n+1}{2}/\binom{w+t}{2}\rfloor$ for large $w$ with respect to $t$, and then show that $C(n,t,w,w-2)=\binom{n+1}{2}/\binom{w+t}{2}$ holds under certain conditions and determine the unique structure when the equality holds. In Section~\ref{sec-(n,t,w,w-2)-small-w}, we first use Tur\'an's theorem to prove $C(n,t,3,1) \leq\lfloor \frac{n(n+1)}{6(t+1)}\rfloor$ for any $n\geq 3(t+1)$, and then use frames to construct $(n,t,3,1)$ and $(n,t,4,2)$-LPECC codes. Next, the general $q$-ary LPECC codes are considered in Section~\ref{sec-asymptotic}, where an asymptotic result is given. In particular, when $q=2$, this determines the asymptotic behavior of binary LPECC codes. Further, the exact values of the maximum sizes of constant-power error-correcting cooling (CPECC) codes are discussed in Section~\ref{sec-cpecc}. Finally, Section~\ref{conclusion} provides some concluding remarks.

\section{Preliminaries}\label{sec-preliminary}

Let $[n]$ denote the set $\{1,2,\ldots,n\}$. For any finite set $X$, let $2^{X}$ denote the power set of $X$, and let $\binom{X}{k}$ denote the family of all $k$-subsets of $X$. Write $A\triangle B$ for the symmetric difference of two sets $A,B\subset X$.  For any family $\A\subset 2^{X}$, we denote $U(\A):=\cup_{A\in\A} A$.

\subsection{Combinatorial characterization of LPECC codes}\label{subsec-comb}

 Let us write $\bm x:=(x_1,\dots,x_n)$ for a codeword in $\{0,1\}^n$. For any  $\bm x\in \{0,1\}^n$, denote by $\supp(\bm x):=\{i\in [n]:x_i\neq 0\}$ the {\it support} of $\bm x$, and denote by $|\bm x|:=|\supp(\bm x)|$ the {\it weight} of $\bm x$. For $\bm x,\bm y\in \{0,1\}^n$, their {\it Hamming distance} is $d(\bm x,\bm y):=|\{i\in [n]:x_i\neq y_i\}|$.

Now let us describe LPECC codes more formally as in \cite{chee2017cooling}. An $(n,t,w,e)$-LPECC code is represented by a code $\C\subset \{0,1\}^n$ with a partition $\C=\C_1\cup\C_2\cup\dots\cup\C_b$, where $b$ is said to be the size of $\C$. Here, each message is encoded to a codeset $\C_i$ of possible codewords instead of exactly one codeword, and each subset $T\subset [n]$  represents a set of hot wires.
In terms of the above terminologies, Property $\bm A(t)$ means that for any $t$-set (hottest wires) $T\subset [n]$ and any $i\in [b]$, there exists $\bm x\in \C_i$ such that $\supp(\bm x)$ does not intersect $T$; Property $\bm B(w)$ means every codeword in $\C$ has weight at most $w$; and Property $\bm C(e)$ means that the Hamming distance between any two codewords from different $\C_i$ must be at least $2e+1$.

Since binary vectors $\bm x\in \{0,1\}^n$ naturally correspond to subsets of an $n$-set, we do not distinguish between binary vectors and subsets. In what follows, for any binary LPECC code $\C=\C_1\cup\C_2\cup\dots\cup\C_b$, we work with its {\it set-theoretic representation} $(X,\B)$ with a partition $\B=\PP_1\cup\PP_2\cup\dots\cup\PP_b$, where $X$ is an underlying $n$-set and $\B,\PP_1,\dots,\PP_b\subset 2^X$ correspond to sets of supports of codewords in $\C,\C_1,\dots,\C_b$, respectively. Members in $\B$ are called blocks.

In \cite{Liu-Ji}, the following combinatorial characterization of an $(n,t,w,e)$-LPECC code is given.
\begin{definition}
Let $X$ be a set of size $n$. We say that $(X,\B)$ is the set-theoretic representation of an $(n,t,w,e)$-LPECC code $\C$ with a partition $\PP_1,\PP_2,\dots,\PP_b$ of $\B\subset 2^X$, where $b$ is the size of $\C$, if all of the following three properties hold:
\begin{itemize}
	\item[(1)] (Property $\bm{A}(t)$) \  For any $T\in \binom{X}{t}$ and any $i\in [b]$, there exists $B\in \mathcal{P}_i$ satisfying $B\cap T =\emptyset$;
	\item[(2)] (Property $\bm{B}(w)$) \  $|B| \leq w$ for any $B\in \mathcal{B}$;
	\item[(3)] (Property $\bm{C}(e)$) \  $|B_i\triangle B_j|\geq 2e+1$, for any $B_i\in \PP_i, B_j\in\PP_j$, $i\neq j$, and $i,j\in [b]$.
\end{itemize}
\end{definition}
Let $C(n,t,w,e)$ be the maximum size of an $(n,t,w,e)$-LPECC code.
When $e=w-2$, in an $(n,t,w,w-2)$-LPECC code of size at least two, Property $\bm{B}(w)$ and Property $\bm{C}(e)$ can be characterized as follows  \cite{Liu-Ji}:
\emph{\begin{itemize}
	\item[(2'')] (Property $\bm{B}(w)$) \  $w-3 \leq |B| \leq w$ for any $B\in \mathcal{B}$;
	\item[(3'')] (Property $\bm{C}(e)$)
	\begin{enumerate}[(i)]
		\item For any pair $\{x,y\}$ of $X$, there exists at most one $i\in [b]$ such that $\{x,y\}$ is contained in some block of $\mathcal{P}_i$;
		\item  There exists at most one $i\in [b]$ such that $\PP_i$ contains blocks of size at most $w-2$. Moreover, if $B\in \PP_i$ and $|B| \leq w-2$, then $B\cap B'=\emptyset$ for any $B'\in \B\setminus \PP_i$;
		\item If $|B_1|=|B_2|=w-1$ with $B_1\in \PP_i, B_2\in\PP_j$ and $i\neq j$, then $B_1\cap B_2=\emptyset$.
	\end{enumerate}
\end{itemize}}

For any $i\in [b]$, let us denote by $\tau(\PP_i)$ the number of pairs $\{x,y\}\in \binom{X}{2}$ contained in some block in $\PP_i$. Note that, by double-counting pairs, (3'')-(i) directly implies that
\begin{equation}\label{eq-double-count-pairs}
	\sum_{i\in [b]} \tau(\PP_i) \leq \binom{n}{2}.
\end{equation}

\subsection{Constructions of LPECC codes from packings}
In \cite{Liu-Ji}, the authors showed that LPECC codes can be constructed from packings.
An $r$-$(n,k,\lambda)$ {\it packing} is a pair $(X,\B)$, where $X$ is a set of $n$ elements and $\B\subset \binom{X}{k}$ such that any $A\in \binom{X}{r}$ is contained in at most $\lambda$ sets in $\B$. Denote by $D_\lambda(n,k,r)$ the packing number, that is, the maximum size of $\B$ in an $r$-$(n,k,\lambda)$ packing $(X,\B)$. If $\lambda=1$, write $D(n,k,r)$ for $D_1(n,k,r)$. It is clear that $D(n,k,r)\leq \binom{n}{r}/\binom{k}{r}$.
An $r$-$(n,k,1)$ packing $(X,\B)$ is called a {\it Steiner system} $S(r,k,n)$ if $|\B|=\binom{n}{r}/\binom{k}{r}$, that is, each $A\in \binom{X}{r}$ is contained in exactly one set in $\B$. When $r=2$, a Steiner system $S(2,k,n)$ is known as a {\it balanced incomplete block design} and usually denoted by $(n,k,1)$-BIBD. For Steiner systems and block designs, one may refer to e.g. \cite{beth1999design,colbourn2007crc} and the references therein.

\begin{theorem}[{\cite[Theorem 4.3]{Liu-Ji}}]\label{thm-packing}
	Let $1\leq t \leq w$ and $w\geq e+1$. Then $C(n,t,w,e) \geq D(n+1,w+t,w-e)$.
\end{theorem}

Let us briefly describe the construction of an LPECC code in Theorem~\ref{thm-packing}.
Given a $(w-e)$-$(n+1,w+t,1)$ packing $(X\cup\{\infty\},\B)$, where $\infty\notin X$ and $|X|=n$. Let $\B'=\{ B\in \B : \infty\notin B\}$ and let $\B''=\{ B\setminus\{\infty\} : \infty\in B\in \B\}$. For each $B\in\B'$, let $\PP_B=\binom{B}{w}$; for each $B\in\B''$, let $\PP_B=\binom{B}{w-1}$. Note that $\bigcup_{B\in \B'\cup \B''} \PP_B$ is a disjoint union by the packing property. Then it is routine to check that $\C=\bigcup_{B\in \B'\cup \B''} \PP_B$ is an $(n,t,w,e)$-LPECC code, so $C(n,t,w,e) \geq D(n+1,w+t,w-e)$.

However, the packing numbers $D(n,k,r)$ have not been determined completely. For near-optimal packings, R\"odl~\cite{Rodl} in 1985 introduced the celebrated `nibble method' to show that, for all fixed $k>r\geq 2$,
\begin{equation}\label{eq-near-optimal-packing}
	\lim_{n\rightarrow \infty} \frac{D(n,k,r)}{\binom{n}{r}/\binom{k}{r}} = 1,
\end{equation}
which resolved a conjecture of Erd\H{o}s and Hanani~\cite{Erdos-Hanani}.
For optimal packings, Wilson~\cite{wilson1972existence1,wilson1972existence2,wilson1975existence} proved that if $n$ is sufficiently large and satisfies necessary divisibility conditions that $\binom{k}{2} \mid \binom{n}{2}$ and $k-1\mid n-1$, then there exists an $(n,k,1)$-BIBD and hence
\begin{equation}\label{eq-optimal-packing}
	D(n,k,2)=\frac{\binom{n}{2}}{\binom{k}{2}}.
\end{equation}
More generally, in 2014, Keevash~\cite{Keevash} proved that there exists a Steiner system $S(r,k,n)$ and hence
$
D(n,k,r)=\binom{n}{r}/\binom{k}{r},
$
provided that $n$ is large enough and satisfies certain necessary divisibility conditions, i.e., $\binom{k-i}{r-i} \mid \binom{n-i}{r-i}$ for each $0\leq i \leq r-1$. % See also \cite{Delcourt-Postle} and \cite{Glock-Kuhn-Lo-Osthus} for alternate proofs of the existence conjecture of combinatorial designs.

\section{An upper bound for $(n,t,w,w-2)$-LPECC codes with large $w$}	\label{sec-(n,t,w,w-2)-large-w}
In this section, we give a tight upper bound on the size of $(n,t,w,w-2)$-LPECC codes for relatively large $w$ with respect to $t$.
\begin{theorem}\label{thm-(n,t,w,w-2)-upper}
For $w\geq t^2+2t+2\geq 5$,	let $\mathcal{C}$ be an $(n,t,w,w-2)$-LPECC code of size $b$ with  $n\geq w+t-1$. Then $b \leq \lfloor \binom{n+1}{2}/\binom{w+t}{2}\rfloor$.
\end{theorem}

\begin{remark}
	The condition that $n\geq w+t-1$ is necessary since otherwise $\lfloor \binom{n+1}{2}/\binom{w+t}{2}\rfloor =0$.
	Moreover, note that, when $t=1$, Theorem~\ref{thm-(n,t,w,w-2)-upper} provides an upper bound $C(n,1,w,w-2)\leq \lfloor \binom{n+1}{2}/\binom{w+1}{2}\rfloor$ with $n\geq w\geq 5$, which complements the bounds for the case $w=3,4$ in {\cite[Theorems 3.2 and 3.3]{Liu-Ji}}.
\end{remark}

To prove Theorem~\ref{thm-(n,t,w,w-2)-upper}, as in \cite{Liu-Ji} one can first consider the case that each block $B$ in the set-theoretic representation of $\C$ has size $w$ or $w-1$. Otherwise we can find an $(n-|B|,t,w,w-2)$-LPECC code of size $b-1$, whose set-theoretic representation contains only blocks of size $w$ and $w-1$ by (3''). See details below.

\begin{lemma}\label{lemma-(n,t,w,w-2)-upper}
	Let $\mathcal{C}$ be an $(n,t,w,w-2)$-LPECC code of size $b$ with  $w\geq t^2+2t+2\geq 5$. Let $(X,\B)$ be the set-theoretic representation of $\C$ with a partition $\{\PP_i: i\in [b]\}$. If $w-1\leq |B| \leq w$ for any $B\in \mathcal{B}$, then we have $b \leq \lfloor \binom{n+1}{2}/\binom{w+t}{2}\rfloor$.
\end{lemma}

\begin{proof}[Proof of Theorem~\ref{thm-(n,t,w,w-2)-upper} assuming Lemma~\ref{lemma-(n,t,w,w-2)-upper}]
	Let $(X,\B)$ be the set-theoretic representation of $\C$ with a partition $\{\PP_i: i\in [b]\}$. If $b=1$ then $b\leq \lfloor \frac{n(n+1)}{w(w+1)}\rfloor$ trivially holds, so we assume $\C$ is of size at least two.
	
	If $w-1\leq |B| \leq w$ for any $B\in \mathcal{B}$, then we are done by Lemma~\ref{lemma-(n,t,w,w-2)-upper}. Otherwise there exists a block $B\in \PP_j$ with $w-3 \leq |B| \leq w-2$ for some $j\in [b]$ by (2''). Moreover, (3'') implies that $\B\setminus \PP_j$ has no blocks of size at most $w-2$ and that $B\cap B'=\emptyset$ for any $B'\in \B\setminus \PP_j$. So deleting $\PP_j$ and all elements in $B$ yields an $(n-|B|,t,w,w-2)$-LPECC code of size $b-1$, whose set-theoretic representation is $(X\setminus B, \B\setminus\PP_j)$ with the partition $\{\PP_i: i\in [b]\setminus\{j\}\}$. Since each block in $\B\setminus \PP_j$ has size $w-1$ or $w$, by Lemma~\ref{lemma-(n,t,w,w-2)-upper},
	\begin{equation}\label{equa:w-3-1}
	b-1
	 \leq  \left\lfloor \frac{\binom{n-|B|+1}{2}}{\binom{w+t}{2}}\right\rfloor
	 \leq  \left\lfloor \frac{\binom{n-w+4}{2}}{\binom{w+t}{2}}\right\rfloor,
	\end{equation}
	where the second inequality is due to $|B| \geq w-3$. Note that
	\begin{equation}\label{equa:w-3-2}
	1+  \frac{\binom{n-w+4}{2}}{\binom{w+t}{2}}
	\leq  \frac{\binom{n+1}{2}}{\binom{w+t}{2}}
\end{equation}
    is equivalent to $n\geq w+t-1+\frac{(t+3)(t+2)}{2(w-3)}$. Hence
    $
    b\leq \lfloor \binom{n+1}{2}/\binom{w+t}{2}\rfloor$ holds when $n\geq w+t-1+\frac{(t+3)(t+2)}{2(w-3)}$. It remains to verify the bound for the case $w+t-1\leq n< w+t-1+\frac{(t+3)(t+2)}{2(w-3)}$. Since we assume that $\C$ has size at least two, by Property $\bm{C}(e)$ we have
    $$
    w+t-1+\frac{(t+3)(t+2)}{2(w-3)} > n\geq |B_i\triangle B_j|\geq 2(w-2)+1 =2w-3,
    $$
    for any $B_i\in \PP_i, B_j\in\PP_j$, $i\neq j$, and $i,j\in [b]$. This implies $2(w-3)(w-t-2)<(t+3)(t+2)$, which holds only when $(w,t)=(5,1)$ since $w\geq t^2+2t+2$. It is easy to check that $b\leq 1$ when $(w,t)=(5,1)$ and $n<8$, completing the proof.
\end{proof}

Now we are going to prove Lemma~\ref{lemma-(n,t,w,w-2)-upper}. % Section~\ref{subsec-comb}	Let us first briefly describe the proof idea.
Our main idea is to double-count the number of pairs $\sum_{i\in [b]} \tau(\PP_i)$, which has an upper bound by Eq.~(\ref{eq-double-count-pairs}). For a lower bound, we will apply Property $\bm{A}(t)$ and estimate $\tau(\PP_i)$ carefully by considering the size of $U(\PP_i)=\cup_{B\in\PP_i} B$.  %In addition, we also need to count the number of covered elements besides the number of pairs.
\begin{proof}[Proof of Lemma~\ref{lemma-(n,t,w,w-2)-upper}]Note that each block in $\PP_i$ has size $w-1$ or $w$. Hence it is easy to see that $|U(\PP_i)|\geq w+t-1$ by Property $\bm{A}(t)$.
  For every $i\in [b]$, let us denote $\PP_i^{w-1}:=\PP_i \cap \binom{X}{w-1}$.
 We first show lower bounds on $\tau(\PP_i)$. In order to analyze $\tau(\PP_i)$ carefully, we consider two cases: $|U(\PP_i^{w-1})|\leq w+t-2$ and $|U(\PP_i^{w-1})|\geq w+t-1$.
\begin{claim}\label{claim-tau-lowerbound1}
	If $|U(\PP_i^{w-1})|\leq w+t-2$, then $\tau(\PP_i)\geq \binom{w+t}{2}$.
\end{claim}
\begin{proof}
	 We divide the proof into three cases.
	\begin{enumerate}
		\item [(a)] If $|U(\PP_i)|=w+t-1$,  Property $\bm{A}(t)$ implies that for each $t$-subset $T\subset U(\PP_i)$ we always have $U(\PP_i)\setminus T \in \PP_i$, since each block in $\PP_i$ has size at least $w-1$. This shows that $\binom{U(\PP_i)}{w-1} \subset \PP_i$ and hence $|U(\PP_i^{w-1})|=|U(\PP_i)|=w+t-1$, contradicting to the assumption that $|U(\PP_i^{w-1})|\leq w+t-2$; so this case does not happen.
		\item [(b)] If $|U(\PP_i)|= w+t$, we claim that $\tau(\PP_i)= \binom{|U(\PP_i)|}{2}= \binom{w+t}{2}$. Indeed, for any pair $\{x,y\}\subset U(\PP_i)$, %$\{x,y\}\in \binom{X}{2}$,
since $|U(\PP_i^{w-1})|\leq w+t-2$, we can choose a $t$-subset $T\subset U(\PP_i)$ such that $\left|U(\PP_i^{w-1})\setminus T\right|\leq w-2$ and $\{x,y\}\not\subset T$.  By the definition of $U(\PP_i^{w-1})$, the inequality $|U(\PP_i^{w-1})\setminus T|\leq w-2$ shows that $\PP_i$ has no $(w-1)$-sets disjoint from $T$. Then it follows from Property $\bm{A}(t)$ that the $w$-set $U(\PP_i)\setminus T$ must be a block in $\PP_i$. Thus $\{x,y\}$ is covered by the block $U(\PP_i)\setminus T$ in $\PP_i$, and hence $\tau(\PP_i)= \binom{|U(\PP_i)|}{2}= \binom{w+t}{2}$.
		\item [(c)] If $|U(\PP_i)|\geq w+t+1$, as in case (b) we can choose a $t$-subset $T\subset U(\PP_i)$ such that $\left|U(\PP_i^{w-1})\setminus T\right|\leq w-2$, and hence there exists some block $B$ of size $w$ in $\PP_i$ with $B\cap T=\emptyset$. Let $Y=U(\PP_i)\setminus B$. Then $|Y|=|U(\PP_i)|-w \geq t+1$ and we have
			\begin{align*}
				\tau(\PP_i) & \geq \binom{|B|}{2} % \left| \{u,v\}: \{u,v\}\in \binom{B}{2} \right|
+ \left|\left\lbrace  \{y,x\}\subset X: y\in Y \text{ and } \{y,x\} \text{ is contained in some block in } \PP_i \right\rbrace \right| \\		
%+ \left|\left\lbrace  \{y,x\}: y\in Y, y\neq x \text{ and } \{y,x\} \text{ is contained in some block in } \PP_i \right\rbrace \right| \\		
				& \geq \binom{w}{2}   + ((w-1)-1) + ((w-1)-2) + \dots +((w-1)-(t+1))  \\
				  & = \frac{1}{2} \left( w^2+(2t+1)w-(t+1)(t+4) \right)  \geq \frac{1}{2} \left( w^2+(2t-1)w+t(t-1) \right) = \binom{w+t}{2},
				\end{align*}
				where we used $w\geq t^2+2t+2$ in the last inequality. Here, the second inequality is obtained greedily by considering each $y\in Y$ in order, and choosing $x$ from a block removing previous $y$'s.
	\end{enumerate}
Therefore, $\tau(\PP_i)\geq \binom{w+t}{2}$ holds provided $|U(\PP_i^{w-1})|\leq w+t-2$.
\end{proof}
Moreover, we can show that $\tau(\PP_i)\geq \binom{w+t-1}{2}$ holds when $U(\PP_i^{w-1})$ has size at least $w+t-1$.
  \begin{claim}\label{claim-tau-lowerbound2}
  	If $|U(\PP_i^{w-1})|\geq w+t-1$, then $\tau(\PP_i)\geq \binom{w+t-1}{2}$.
  \end{claim}
\begin{proof}
	   Since $|U(\PP_i^{w-1})|\geq w+t-1>0$, there exists some block $B$ of size $w-1$ in $\PP_i$. The proof can be divided into two cases.
  \begin{enumerate}[(a)]
		\item  If $|U(\PP_i)|=w+t-1$,  Property $\bm{A}(t)$ implies that for each $t$-subset $T\subset U(\PP_i)$ we have $U(\PP_i)\setminus T \in \PP_i$, since each block in $\PP_i$ has size at least $w-1$. Hence $\binom{U(\PP_i)}{w-1} \subset \PP_i$ and each pair $\{x,y\}$ in $U(\PP_i)$ can be covered by some block in $\PP_i$. Thus $\tau(\PP_i)= \binom{|U(\PP_i)|}{2}= \binom{w+t-1}{2}$.
		\item  If $|U(\PP_i)|\geq w+t$, then $Y=U(\PP_i)\setminus B$ has size at least $|U(\PP_i)|-(w-1) \geq t+1$.  Then
		\begin{align*}
		\tau(\PP_i) & \geq \binom{|B|}{2} + \left|\left\lbrace  \{y,x\}\subset X: y\in Y \text{ and } \{y,x\} \text{ is contained in some block in } \PP_i\right\rbrace  \right| \\				
		& \geq \binom{w-1}{2}   + \left( (w-1)-1\right)  + \left( (w-1)-2\right)  + \dots +\left( (w-1)-(t+1)\right)   \\
		& = \frac{1}{2} \left( w^2+(2t-1)w+2-(t+1)(t+4) \right)  > \frac{1}{2} \left( w^2+(2t-3)w+(t-1)(t-2) \right) = \binom{w+t-1}{2}.
		\end{align*}
 \end{enumerate}
Therefore, it holds that $\tau(\PP_i)> \binom{w+t-1}{2}$ when $|U(\PP_i^{w-1})|\geq w+t-1$.
\end{proof}

Let $S_1=\{i\in [b] : |U(\PP_i^{w-1})|\leq w+t-2\}$, $S_2=\{i\in [b]: |U(\PP_i^{w-1})|\geq w+t-1 \}$, and let $s_1=|S_1|,s_2=|S_2|$.
Combining Claim~\ref{claim-tau-lowerbound1} and Claim~\ref{claim-tau-lowerbound2}, we have shown that, for each $i\in [b]$,
\begin{itemize}
	\item if $i\in S_1$, then $\tau(\PP_i)\geq \binom{w+t}{2}$;
	\item if $i\in S_2$, then $\tau(\PP_i)\geq \binom{w+t-1}{2}$.
\end{itemize}
Then, by Eq. \eqref{eq-double-count-pairs}, we see that
\begin{equation}\label{equa-count-pair1}
s_1\binom{w+t}{2} + s_2\binom{w+t-1}{2} \leq \sum_{i\in S_1} \tau(\PP_i)+\sum_{i\in S_2} \tau(\PP_i)= \sum_{i\in [b]} \tau(\PP_i) \leq \binom{n}{2}.
\end{equation}
It follows that
\begin{equation}\label{equa-count-pair-equiv}
 b\cdot \binom{w+t}{2} =(s_1+s_2)\binom{w+t}{2} \leq \binom{n}{2}+s_2 (w+t-1).
\end{equation}
 Thus it suffices to show that $s_2(w+t-1)\leq n$. By (3'')-(iii), all sets in $\{U(\PP_i^{w-1}) : i\in S_2\}$ are pairwise disjoint. Since $U(\PP_i^{w-1})$ has size at least $w+t-1$ for every $i\in S_2$, we have
 \begin{equation}\label{equa-n-union}
 n\geq \left| \bigcup_{i\in S_2} U(\PP_i^{w-1}) \right| = \sum_{i\in S_2} \left| U(\PP_i^{w-1})\right|  \geq s_2 (w+t-1)
\end{equation}
 as desired. This completes the proof.
\end{proof}

Let $w,t\geq 1$ be fixed with $w\geq t^2+2t+2$. Combining Theorem~\ref{thm-(n,t,w,w-2)-upper}, Theorem~\ref{thm-packing} and \eqref{eq-near-optimal-packing} yields
\[
\lim_{n\rightarrow \infty} \frac{C(n,t,w,w-2)}{\binom{n+1}{2}/\binom{w+t}{2}} = 1.
\]
This shows that the upper bound in Theorem~\ref{thm-(n,t,w,w-2)-upper} is asymptotically sharp. Moreover, combining Theorem~\ref{thm-(n,t,w,w-2)-upper}, Theorem~\ref{thm-packing} and \eqref{eq-optimal-packing}, we see that the upper bound in Theorem~\ref{thm-(n,t,w,w-2)-upper} is indeed optimal for many cases.
\begin{corollary}\label{cor-value-(n,t,w,w-2)}
	Let $w\geq t^2+2t+2\geq 5$.
	Then it holds that
	\[
	C(n,t,w,w-2)=\frac{\binom{n+1}{2}}{\binom{w+t}{2}},
	\]
	if $n$ is sufficiently large and satisfies $\binom{w+t}{2} \mid \binom{n+1}{2}$ and $w+t-1\mid n$.
\end{corollary}

 Next, we show the extremal case: any $(n,t,w,w-2)$-LPECC code of size $\binom{n+1}{2}/\binom{w+t}{2}>1$ must come from an $(n+1,w+t,1)$-BIBD, provided that $w> t^2+2t+2$ and $\binom{w+t}{2} \mid \binom{n+1}{2}$.

\begin{theorem}\label{thm-BIBD}
	Let $w> t^2+2t+2$, and suppose that $n$ satisfies $\binom{w+t}{2} \mid \binom{n+1}{2}$. If there exists an $(n,t,w,w-2)$-LPECC code $\mathcal{C}$ with size $b=\binom{n+1}{2}/\binom{w+t}{2}>1$, and  $(X,\B)$ is the set-theoretic representation of $\C$ with a partition $\{\PP_i: i\in [b]\}$, then
	\begin{itemize}
		\item[\romannumeral1)] each block in $\B$ has size $w-1$ or $w$;
		\item[\romannumeral2)] $w+t-1\mid n$;
		\item[\romannumeral3)] for each $i\in [b]$, either $|U(\PP_i)|=w+t-1$ and $\binom{U(\PP_i)}{w-1} \subset\PP_i$, or $|U(\PP_i)|=w+t$ and $\binom{U(\PP_i)}{w}=\PP_i$;
		\item[\romannumeral4)] $(X\cup\{\infty\},\mathcal{D})$ forms an $(n+1,w+t,1)$-BIBD, where
		\[
		\mathcal{D} = \{U(\PP_i)\cup\{\infty\} : |U(\PP_i)|=w+t-1\} \cup \{U(\PP_i) : |U(\PP_i)|=w+t\}.
		\]
	\end{itemize}
\end{theorem}

\begin{remark}
	Theorem~\ref{thm-BIBD} indicates that for $w> t^2+2t+2$, if $\binom{w+t}{2} \mid \binom{n+1}{2}$ but $w+t-1\nmid n$, then the size $b=\binom{n+1}{2}/\binom{w+t}{2}>1$ cannot be attained by any $(n,t,w,w-2)$-LPECC code, so the upper bound in Theorem~\ref{thm-(n,t,w,w-2)-upper} can be improved a little bit for this case.
\end{remark}

 \begin{proof}
 	 As in the proof of Theorem~\ref{thm-(n,t,w,w-2)-upper}, we first show that every block $B\in\B$ must be of size either $w$ or $w-1$. Suppose on the contrary that there exists a block $B\in \PP_j$ with $w-3 \leq |B| \leq w-2$ for some $j\in [b]$ by (2''). As before, deleting $\PP_j$ and all elements in $B$ gives an $(n-|B|,t,w,w-2)$-LPECC code of size $b-1$. Then all the equalities in \eqref{equa:w-3-1} and \eqref{equa:w-3-2} must hold since $\C$ has size $b=\binom{n+1}{2}/\binom{w+t}{2}>1$. Note that the equality in \eqref{equa:w-3-2} is equivalent to $n= w+t-1+\frac{(t+3)(t+2)}{2(w-3)}$. However, it follows from $w> t^2+2t+2$ that
 	 $$
 	 w+t-1 <n= w+t-1+\frac{(t+3)(t+2)}{2(w-3)}
 	  \leq w+t-1+\frac{(t+3)(t+2)}{2(t^2+2t)}\leq w+t+1,
 	 $$
 	 so $n\in\{w+t,w+t+1\}$. Then $\binom{w+t}{2} \mid \binom{n+1}{2}$ implies that $\binom{n+1}{2}$ is divisible by one of $\binom{n}{2}$ and $\binom{n-1}{2}$, which is impossible since $n\geq 2w-3\geq 2t^2+4t+1> 5$ by Property $\bm C(e)$. This shows that each block in $\B$ must be of size either $w$ or $w-1$.
 	
 	  It remains to show that the equality in Lemma~\ref{lemma-(n,t,w,w-2)-upper} implies properties \romannumeral2)--\romannumeral4) in Theorem~\ref{thm-BIBD}. Note that all equalities in \eqref{equa-count-pair1},\eqref{equa-count-pair-equiv} and \eqref{equa-n-union} should hold. The equality in \eqref{equa-n-union} implies $n= \left| \cup_{i\in S_2} U(\PP_i^{w-1}) \right| = s_2 (w+t-1)$, so $s_2=n/(w+t-1)$ and \romannumeral2) is verified. This also shows that $\cup_{i\in S_2} U(\PP_i^{w-1})=X$ and hence $\PP_i$ contains only blocks of size $w$ for each $i\in S_1$ by (3'')-(iii). Moreover, the equality in \eqref{equa-count-pair1} requires $\tau(\PP_i)= \binom{w+t}{2}$ for $i\in S_1$, and $\tau(\PP_i)= \binom{w+t-1}{2}$ for $i\in S_2$.
 	  When the equality in Claim~\ref{claim-tau-lowerbound1} holds, one can observe that  case (c) in the proof of Claim~\ref{claim-tau-lowerbound1} cannot happen since $w> t^2+2t+2$. Hence the proof of Claims~\ref{claim-tau-lowerbound1} implies $|U(\PP_i)|= w+t$ and $\binom{U(\PP_i)}{w} = \PP_i$ for each $i\in S_1$. Similarly, the proof of Claim~\ref{claim-tau-lowerbound2} implies $|U(\PP_i)|= w+t-1$ and $\binom{U(\PP_i)}{w-1} \subset \PP_i$ for each $i\in S_2$.  Then properties \romannumeral3) and  \romannumeral4) follow.
 \end{proof}

\section{Bounds for $(n,t,w,w-2)$-LPECC codes with small $w$}\label{sec-(n,t,w,w-2)-small-w}
In the last section, we have given a tight upper bound for $(n,t,w,w-2)$-LPECC codes for sufficiently large $w$ with respect to $t$. However, that bound does not apply to small weight $w$ with respect to $t$. For example, in \cite{Liu-Ji}, it was shown that $C(n,2,3,1)=\lfloor \frac{n(n+1)}{18} \rfloor$ for any positive even integer $n$ with $n\neq 6$. In this section, we consider bounds for $(n,t,w,w-2)$-LPECC codes with small $w\in\{3,4\}$ and show that $C(n,t,3,1) =\lfloor \frac{n(n+1)}{6(t+1)}\rfloor$ for any positive integer $t$ and any even $n\neq 2$ with $n\equiv 2 \pmod{3(t+1)}$. Thus $\lim_{n\to\infty}  C(n,t,3,1)/n^2 =\frac{1}{6(t+1)}$. Moreover, we have $\liminf_{n\to\infty} C(n,t,4,2)/n^2 \geq \frac{1}{12(t+1)}$, which is much larger than the upper bound in Theorem~\ref{thm-(n,t,w,w-2)-upper} when $t>5$. This lower bound might be best possible for $t>5$.  Unfortunately, our proof method for the upper bound of $C(n,t,3,1)$ does not work for that of $C(n,t,4,2)$, so new ideas are needed to determine the asymptotic or exact value of $C(n,t,4,2)$.

\subsection{An upper bound for $(n,t,3,1)$-LPECC codes}
In this subsection, we give the following sharp upper bound for $(n,t,3,1)$-LPECC codes.
\begin{theorem}\label{thm-(n,t,3,1)-upper}
	If $\mathcal{C}$ is an $(n,t,3,1)$-LPECC code of size $b$ with $n\geq 3(t+1)$, then $b \leq  \lfloor\frac{n(n+1)}{6(t+1)}\rfloor$.
\end{theorem}

Let us first prepare some lemmas below. Here, we denote by $e(G)$ the number of edges in the graph $G$.
\begin{lemma}\label{lem:transversal-edge-lower}
	 Let $G$ be a graph on $n$ vertices. Suppose that, for any vertex set $W\in \binom{V(G)}{n-t}$, the induced subgraph $G[W]$ contains at least one triangle. Then  $e(G)\geq 3(t+1)$.
\end{lemma}
\begin{proof}
  We prove by induction on $t$. The base case $t=0$ is trivial. Suppose that Lemma~\ref{lem:transversal-edge-lower} holds for $t-1$ or less and consider the case $t$. Let $G$ be an $n$-vertex graph satisfying the above assumption for $t$.% that, for any $W\in \binom{V(G)}{n-t}$, the induced subgraph $G[W]$ contains at least one triangle.	

   If every vertex in $G$ has degree at most $2$, then $G$ is a vertex-disjoint union of some isolated vertices, paths and cycles. In particular, all triangles in $G$ are pairwise vertex-disjoint. Since $G[W]$ contains at least one triangle for any $W\in \binom{V(G)}{n-t}$, $G$ must contain at least $t+1$ vertex-disjoint triangles.  Otherwise one can delete  $t$ vertices from $G$ (one vertex from each triangle) to get a $W$ such that $G[W]$ is triangle-free.
  Hence $e(G)\geq 3(t+1)$.

  Now assume that the maximum degree of $G$ is at least $3$.  Let us delete a vertex in $G$ of degree at least $3$. Note that the resulting $(n-1)$-vertex graph $G'$ satisfies that, for any $W\in \binom{V(G')}{(n-1)-(t-1)}$, the induced subgraph $G'[W]$ contains at least one triangle. By the induction hypothesis, we have $e(G')\geq 3t$. Therefore, $e(G)\geq e(G')+3\geq 3(t+1)$. This completes the proof.
\end{proof}

For our application, given a $3$-uniform family $\mathcal{F}\subset \binom{[n]}{3}$ satisfying Property $\bm A(t)$, we will simply consider the $n$-vertex graph that is the union of triangles corresponding to $3$-sets in $\mathcal{F}$. For such graphs, clearly the condition in Lemma~\ref{lem:transversal-edge-lower} holds.
We remark that graphs consisting of the union of some triangles have been studied by Erd\H{o}s, Gallai and Tuza~\cite{erdos1996covering} under the name of {\it ``triangular graphs''}.

 The next result that we need is the celebrated Tur\'an's theorem.
 \begin{lemma}[Tur\'an~\cite{turan1941external}]\label{lem-turan}
 	If an $n$-vertex graph $G$ is $K_{r+1}$-free, that is, $G$ contains no clique of size $r+1$, then
 	\[
 	e(G) \leq e(T(n,r)),
 	\]
 	where the Tur\'an graph $T(n,r)$ stands for the $n$-vertex complete balanced $r$-partite graph.
 	
 \end{lemma}
 We also need the following estimate.
 We will use it to give a  lower bound on the sum of the numbers of vertices and edges in the complement of a Tur\'an graph $T(v,v-m)$.
\begin{lemma}\label{lem-estimate}
	Let $v,m$ be two positive integers with $v>m$. Then we have
	\[
	v+ \lambda \binom{\lceil \frac{v}{v-m}\rceil}{2} + (v-m-\lambda) \binom{\lfloor \frac{v}{v-m}\rfloor}{2}
	\geq 3m,
	\]
	where $\lambda\in [0,v-m-1]$ is an integer satisfying $\lambda \equiv v \pmod{v-m}$.
\end{lemma}
 \begin{proof}
 	
 	Let $s=\lfloor \frac{v}{v-m}\rfloor$. Then $v=(v-m)s+\lambda$ by definition. Note that $\lambda=0$ if and only if $(v-m)\mid v$. So
 	\[
v+ \lambda \binom{\lceil \frac{v}{v-m}\rceil}{2} + (v-m-\lambda) \binom{\lfloor \frac{v}{v-m}\rfloor}{2}
 	= v+ \lambda \binom{s+1}{2} + (v-m-\lambda) \binom{s}{2} \triangleq\xi.
 	\]
 	We split the proof  into three cases.
 	
 	{\bf Case 1.} $v> 2m$.
 	Then $s=1$ and $\lambda=m$.
 %$1<\frac{v}{v-m}<2$; hence $s=1$ and $\lambda=v-(v-m)s=m$.
 It follows that
 	%v+ \lambda \binom{s+1}{2} + (v-m-\lambda) \binom{s}{2}
 	\[
    \xi= v+\lambda = v+m \geq 3m.
 	\]
 	
 	{\bf Case 2.} $3m/2 < v\leq 2m$.  Then $s=2$ and $\lambda=2m-v$.
 %it is easy to see that $2\leq \frac{v}{v-m}<3$. Thus $s=2$ and $\lambda=v-(v-m)s=2m-v$.
 Therefore,
 	 %v+ \lambda \binom{s+1}{2} + (v-m-\lambda) \binom{s}{2}
 	 \[
 	 \xi = v+(2m-v)\cdot 3 + (v-m-(2m-v))=  3m.
 	 \]
 	
 	{\bf Case 3.} $m < v\leq 3m/2$. Note that $\frac{v}{v-m}\geq 3$ and hence $s\geq 3$.
 %Since $\binom{x}{2}:=\frac{x(x-1)}{2}$ is a convex real function for $x\geq 2$,
 By Jensen's inequality, we have
 	\begin{align*}
 	\xi =v+ \lambda \binom{s+1}{2} + (v-m-\lambda) \binom{s}{2}
 	&\geq v+ (v-m) \binom{\frac{\lambda}{v-m}(s+1) + \frac{v-m-\lambda}{v-m}s}{2}\\
 	&= v+ (v-m) \binom{\frac{v}{v-m}}{2} = v+\frac{vm}{2(v-m)} = v+ \frac{m}{2} + \frac{m^2}{2(v-m)}.
 	\end{align*}
 	Note that $f(x)=x+\frac{m^2}{2x}$ is decreasing in $(0,m/\sqrt{2})$. Since $v-m\leq m/2<m/\sqrt{2}$,
 %	v+ \lambda \binom{s+1}{2} + (v-m-\lambda) \binom{s}{2}
% 	\geq v+ \frac{m}{2} + \frac{m^2}{2(v-m)}
    \[
 	\xi= \frac{3m}{2} +(v-m)+ \frac{m^2}{2(v-m)}
 	\geq 3m.
 	\]
 	This finishes the proof of Lemma~\ref{lem-estimate}.
 \end{proof}

 Now we prove Theorem~\ref{thm-(n,t,3,1)-upper}.

 \begin{proof}[Proof of Theorem~\ref{thm-(n,t,3,1)-upper}]
 	Let $(X,\B)$ be the set-theoretic representation of $\C$ with a partition $\{\PP_i: i\in [b]\}$. By Property $\bm B(w)$, each block in $\B$ has size at most three.
 	
  %For each $i\in [b]$, we still denote by $\tau(\PP_i)$ the number of pairs $\{x,y\}\in \binom{X}{2}$ covered by some block in $\PP_i$.
  If $\PP_i$ has only blocks of size three, we have the following claim for $\tau(\PP_i)$.
 	 \begin{claim}\label{claim-tau-lowerbound-w=3}
 	 	If $\PP_i$ contains only blocks of size three, then we have $\tau(\PP_i)\geq 3(t+1)$.
 	 \end{claim}
 	 \begin{proof}
 	  Consider the graph $G$ whose vertex set is $X$ and edge set consists of all pairs $\{x,y\}\in \binom{X}{2}$ covered by some block in $\PP_i$. It follows that $e(G)=\tau(\PP_i)$. By Property $\bm A(t)$, for any vertex set $W\in \binom{V(G)}{n-t}$, the induced subgraph $G[W]$ must contain at least one triangle. Then Lemma~\ref{lem:transversal-edge-lower} implies that $\tau(\PP_i)\geq 3(t+1)$.
 	 \end{proof}
 	 By using (3'')-(iii) in Property $\bm{C}(e)$, it is not difficult to show that there exists a subset $I\subset [b]$ of size at least $b-\frac{n+1}{2}$, such that every block in $\cup_{i\in I} \PP_i$ has size three.
 	 By Eq. \eqref{eq-double-count-pairs} and Claim~\ref{claim-tau-lowerbound-w=3}, we see that
 	 \[
 	  3(t+1)\cdot |I|\leq \sum_{i\in I} \tau(\PP_i) \leq \binom{n}{2}.
 	 \]
 	 Therefore, we have $b\leq |I|+\frac{n+1}{2} \leq \frac{n(n-1)}{6(t+1)} +\frac{n+1}{2}$. This already shows that $b \leq  \frac{n(n+1)}{6(t+1)} +O(n)$.
 	 Next, by analyzing those $\PP_i$'s containing some block of size two more carefully, we will prove $b \leq  \frac{n(n+1)}{6(t+1)}$.
 	
 	 If each block in $\PP_i$ has size either two or three, let $\PP_i^2:=\PP_i \cap \binom{X}{2}$ and let $v_2(\PP_i)$ denote the size $|U(\PP_i^2)|$. In other words, $v_2(\PP_i)$ is the size of the union of $2$-sets in $\PP_i$. Let
 	 $\PP_i^{3,1}:=\{A\in \PP_i\cap \binom{X}{3} : A \text{ contains no block in $\PP_i^2$}\}$.
 	 \begin{claim}\label{claim-tau+v-lowerbound-w=2,3}
 	 	If every block in $\PP_i$ is of size two or three, then $\tau(\PP_i)+v_2(\PP_i)\geq 3(t+1)$.
 	 \end{claim}
 	 \begin{proof}
 	 	Let $t_i$ be the minimum number such that there exists a $(t_i+1)$-set of $X$ that intersects all sets in $\PP_i^{3,1}$. Let $R\subset X$ be such a $(t_i+1)$-set. 	
		Note that by the minimality of $t_i$, for every $t_i$-set $T\subset X$, there exists some block $A\in\PP_i^{3,1}$ such that $A\cap T=\emptyset$. Hence, replacing $\PP_i$ with $\PP_i^{3,1}$ yields an $(n,\min\{t_i,t\},3,1)$-LPECC code, whose set-theoretic representation is $(X, (\B\setminus\PP_i)\cup \PP_i^{3,1})$. In this new code, the corresponding partition is $\{\PP'_j: j\in [b]\}$ where $\PP'_i:=\PP_i^{3,1}$ and $\PP'_j:=\PP_j$ for $j\neq i$.
 By Claim~\ref{claim-tau-lowerbound-w=3}, $\tau(\PP'_i)\geq 3(\min\{t_i,t\}+1)$. So, if $t_i\geq t$, then we are done since $\tau(\PP_i)+v_2(\PP_i)\geq \tau(\PP_i)\geq \tau(\PP'_i)$. Therefore, we can assume now that $t_i< t$. Moreover, observe that $t_i< t$ implies $\PP_i^2\neq \emptyset$.
 	 	
 	 	Consider the auxiliary graph $H$ whose vertex set is $V(H)=U(\PP_i^2)$ and edge set is $E(H)=\PP_i^2$. Note that $|V(H)|=v_2(\PP_i)$.
 	 	By Property $\bm A(t)$, for any $(t-t_i-1)$-set $S\subset X$, there exists some $B\in \mathcal{P}_i$ such that $B\cap (S\cup R) =\emptyset$. Since $R$ intersects all sets in $\PP_i^{3,1}$, we have $B\in \PP_i\setminus \PP_i^{3,1}$, which means that $B$ either belongs to $\PP_i^2$ or contains some block $B_1\in\PP_i^2$ by the definition of $\PP_i^{3,1}$. Thus we can conclude that, for any $(t-t_i-1)$-set $S\subset X$, there exists an edge $B'\in \mathcal{P}_i^2$ such that $B'\cap S =\emptyset$, where one can let $B'=B$ if $B\in \PP_i^2$ and let $B'=B_1$ if $B$ contains the block $B_1\in \PP_i^2$.
 	 	This implies that $|V(H)|=v_2(\PP_i)\geq (t-t_i-1)+2=t-t_i+1$, and that any $(|V(H)|-(t-t_i-1))$-subset of $V(H)$ is not an independent set in the graph $H$. Consequently, the independence number of $H$ is at most $|V(H)|-(t-t_i-1)-1=v_2(\PP_i)-(t-t_i)\triangleq r$, i.e., the complement graph $\overline{H}$ is $K_{r+1}$-free.

        By Tur\'an's theorem (Lemma~\ref{lem-turan}), we have   $e(\overline{H}) \leq e(T(v_2(\PP_i),r))$. Hence,
         	 	\[
 	 	e(H) \geq \binom{v_2(\PP_i)}{2} - e(T(v_2(\PP_i),r))
 	 	 = \lambda \binom{\lceil \frac{v_2(\PP_i)}{r}\rceil}{2} + (r-\lambda) \binom{\lfloor \frac{v_2(\PP_i)}{r}\rfloor}{2},
 	 	\]
% 	 	\[
% 	 	e(H) \geq \binom{v_2(\PP_i)}{2} - e(T(v_2(\PP_i),r))
% 	 	 = \lambda \binom{\lceil \frac{v_2(\PP_i)}{v_2(\PP_i)-(t-t_i)}\rceil}{2} + (v_2(\PP_i)-(t-t_i)-\lambda) \binom{\lfloor \frac{v_2(\PP_i)}{v_2(\PP_i)-(t-t_i)}\rfloor}{2},
% 	 	\]
 	 	%where $\lambda\in [0,v_2(\PP_i)-(t-t_i)-1]$ is an integer satisfying $\lambda \equiv v_2(\PP_i) \pmod{v_2(\PP_i)-(t-t_i)}$.
 where $\lambda\in [0,r-1]$ is an integer satisfying $\lambda \equiv v_2(\PP_i) \pmod{r}$.
 	 	Thus, by applying Lemma~\ref{lem-estimate} with $v=v_2(\PP_i)$ and $m=t-t_i$, we get $e(H)+v_2(\PP_i) \geq 3(t-t_i)$. By the definition of $\PP_i^{3,1}$, the number of pairs $\{x,y\}\in \binom{X}{2}$ contained in some block in $\PP_i$ is at least the sum of $e(H)$ and the number of pairs contained in some block in $\PP_i^{3,1}=\PP_i'$, so $\tau(\PP_i)\geq\tau(\PP'_i)+e(H)$. Recall that $\tau(\PP'_i)\geq 3(\min\{t_i,t\}+1)=3(t_i+1)$.  Then
 	 	\begin{align*}
 	 	\tau(\PP_i)+v_2(\PP_i)
 	 	 &\geq \tau(\PP'_i)+e(H)+v_2(\PP_i) \\
 	 	  & \geq 3(t_i+1) + 3(t-t_i) =3(t+1). \qedhere
 	 \end{align*}
 	 \end{proof}
 	
 	 Let us finish the proof of Theorem~\ref{thm-(n,t,3,1)-upper}. We first consider the case that for any parameters $n,t$,  each block in $\B$ has size two or three. %then $b \leq  \frac{n(n+1)}{6(t+1)}$ holds. Suppose that each block in $\B$ has size two or three.
 Note that in this case, (3'')-(iii) in Property $\bm C(e)$ tells us that, if $|B_1|=|B_2|=2$ with $B_1\in \PP_i, B_2\in\PP_j$ and $i\neq j$, then $B_1\cap B_2=\emptyset$. This implies that
 	 \[
 	  \sum_{i\in [b]} v_2(\PP_i) \leq n.
 	 \]
 	 Therefore, by applying Claim~\ref{claim-tau+v-lowerbound-w=2,3} and \eqref{eq-double-count-pairs}, we deduce that
 	 \[
 	   3(t+1)b \leq \sum_{i\in [b]} \tau(\PP_i) +\sum_{i\in [b]} v_2(\PP_i) \leq \binom{n}{2}+n = \frac{n(n+1)}{2}.
 	 \]
 	 Consequently, $b \leq  \frac{n(n+1)}{6(t+1)}$ for all $n,t$ whenever each block in $\B$ has size two or three.
 	
 	 Suppose  that $\B$ contains a block of size at most one. By (3'')-(ii) in Property $\bm C(e)$, there exists at most one $i\in [b]$ such that $\PP_i$ contains some block $B$ of size no more than $1$. Note that deleting $\PP_j$ and the element in $B$ yields an $(n-|B|,t,3,1)$-LPECC code of size $b-1$, whose set-theoretic representation is $(X\setminus B, \B\setminus\PP_j)$ with the partition $\{\PP_i: i\in [b]\setminus\{j\}\}$. If $|B|=1$, then every block in $\B\setminus\PP_j$ has size two or three, and hence
 	 \[
 	 b-1 \leq \frac{(n-|B|)(n-|B|+1)}{6(t+1)}= \frac{n(n-1)}{6(t+1)}.
 	 \]
 	 Thus $b\leq 1+\frac{n(n-1)}{6(t+1)}\leq \frac{n(n+1)}{6(t+1)}$ since $n\geq 3(t+1)$. If $|B|=0$, then every block in $\B\setminus\PP_j$ has size three by Property $\bm C(e)$. By Claim~\ref{claim-tau-lowerbound-w=3} and \eqref{eq-double-count-pairs}, we still have $b-1\leq \binom{n}{2}/(3(t+1))=\frac{n(n-1)}{6(t+1)}$ and hence
 	 \[
 	 b\leq 1+\frac{n(n-1)}{6(t+1)}\leq \frac{n(n+1)}{6(t+1)},
 	 \]
 	 completing the proof.
 \end{proof}

\subsection{Lower bounds for $(n,t,3,1)$ and $(n,t,4,2)$-LPECC codes}
Motivated by \cite{Liu-Ji}, in this subsection we use frames to construct LPECC codes. Let us first introduce the definition of frames and state related results.

Let $k, n$ and $g$ be positive integers. A {\it group divisible design} of type $g^n$ and block size $k$, denoted by $k$-GDD, is a triple $(X, \G, \B)$ where $X$ is a set of $ng$ elements, $\G$ is a partition of $X$ into $n$ parts (groups) each of size $g$, and $\B$ is a family of $k$-subsets (blocks) of $X$ satisfying that every pair of elements of X occurs in exactly one group or one block, but not both.
A family $\mathcal{F}$ of $k$-subsets of $X$ is called a {\it partial resolution class} if every element of $X$ occurs in at most one block of $\mathcal{F}$. The elements of X not occurring in the partial resolution class form the complement of the class.
A $k$-GDD $(X, \G, \B)$ is a {\it $k$-frame} if the family $\B$ can be partitioned into a collection $\PP$ of partial resolution classes of $X$ such that the complement of every partial resolution class of $\PP$ is a certain group $G \in\G$. The partial resolution class of $\PP$ with complement $G \in\G$ is said to be a {\it holey parallel class with a hole} $G$.
\begin{lemma}[see, e.g. {\cite[Theorem \uppercase\expandafter{\romannumeral4}.5.30, page 263]{colbourn2007crc}}]\label{lem-3-frame}
	There exists a $3$-frame of type $g^n$ if and only if $n\geq 4$, $g\equiv 0 \pmod2$ and $g(n-1) \equiv 0 \pmod3$.
\end{lemma}

\begin{lemma}[see, e.g. {\cite[Theorem \uppercase\expandafter{\romannumeral4}.5.31, page 263]{colbourn2007crc}}]\label{lem-4-frame}
	There exists a $4$-frame of type $g^n$ if and only if $n\geq 5$, $g\equiv 0 \pmod3$ and $g(n-1) \equiv 0 \pmod4$, except possibly where
	\begin{enumerate}
		\item $g = 36$ and $n = 12$, and
		\item $g\equiv 6 \pmod{12}$ and
		\begin{enumerate}
			\item[(a)] $g=6$ and $n\in \{7, 23, 27, 35, 39, 47 \}$;
			\item[(b)] $g=30$ or $g\in [66,2190]$, and $n\in \{7, 23, 27, 39, 47\}$;
			\item[(c)] $g\in\{42, 54\}\cup [2202,11238]$ and $n\in \{23, 27\}$;
			\item[(d)] $g = 18$ and $n \in \{15, 23, 27\}$.
		\end{enumerate}
	\end{enumerate}
\end{lemma}

Liu and Ji~{\cite[Theorem 4.19]{Liu-Ji}} proved that $C(n,2,3,1)=\lfloor \frac{n(n+1)}{18}\rfloor$ for any even integer $n\geq 4$ with $n\neq 6$, and $C(6,2,3,1)=1$.
Following \cite{Liu-Ji}, we construct $(n,t,3,1)$-LPECC codes and $(n,t,4,2)$-LPECC codes from frames, and hence obtain lower bounds for $C(n,t,3,1)$ and $C(n,t,4,2)$.

\begin{theorem}\label{thm-(n,t,3,1)-lower}
	For any $n \equiv 2 \pmod6$, $C(n,t,3,1)\geq \lfloor\frac{n-2}{3(t+1)}\rfloor \cdot \frac{n}{2} + \lfloor\frac{n}{2(t+1)}\rfloor $. 		
	Therefore, we have that $\liminf_{n\to\infty} C(n,t,3,1)/n^2 \geq \frac{1}{6(t+1)}$.
\end{theorem}

\begin{proof}
	By Lemma~\ref{lem-3-frame}, there exists a $3$-frame $(X,\mathcal{G},\mathcal{B})$ of type $2^{\frac{n}{2}}$. It is easy to see that there are exactly $\frac{n}{2}$ holey parallel classes in a $3$-frame of type $2^{\frac{n}{2}}$. Let $\B_1,\B_2,\dots,\B_{n/2}$ be all holey parallel classes of the $3$-frame $(X,\mathcal{G},\mathcal{B})$. Write $\B_i=\{B_{i,1},B_{i,2},\dots,B_{i,\frac{n-2}{3}}\}$ for each $i\in [n/2]$ and  $\G=\{G_1,\dots,G_{n/2}\}$.

	To get the desired LPECC code,  we consider disjoint collections of  $(t+1)$ sets from each of  $\G$ and  $\B_i$. Formally, set
%  will let sets in $\G \cup \bigcup_{i\in[n/2]} \B_i$ be blocks of our LPECC code and find a suitable partition.
%	Let us set
\[
	\PP_j=\{G_{(t+1)j-t},G_{(t+1)j-t+1},\dots,G_{(t+1)j}\},
	\] for each $j\in [\lfloor\frac{n}{2(t+1)}\rfloor]$, and set \[
	\Q_{i,k}=\{B_{i,(t+1)k-t},B_{i,(t+1)k-t+1},\dots,B_{i,(t+1)k}\},
	\] for every $i\in [n/2]$ and $k\in [\lfloor\frac{n-2}{3(t+1)}\rfloor]$.
	Then $
	\C := \left( \bigcup_{j\in [\lfloor\frac{n}{2(t+1)}\rfloor]} \PP_j \right)  \cup \left( \bigcup_{i\in [\frac{n}{2}], k\in [\lfloor\frac{n-2}{3(t+1)}\rfloor]} \Q_{i,k}\right)\subset \G \cup \bigcup_{i\in[n/2]} \B_i
	$
	 gives a partition with $\lfloor\frac{n-2}{3(t+1)}\rfloor \cdot \frac{n}{2} + \lfloor\frac{n}{2(t+1)}\rfloor$ parts. One can directly check by definition that $\C$ is an $(n,t,3,1)$-LPECC code.  Thus $C(n,t,3,1)\geq \lfloor\frac{n-2}{3(t+1)}\rfloor \cdot \frac{n}{2} + \lfloor\frac{n}{2(t+1)}\rfloor $.
\end{proof}

Combining Theorem~\ref{thm-(n,t,3,1)-upper} and Theorem~\ref{thm-(n,t,3,1)-lower}, we immediately derive the asymptotic behavior of $C(n,t,3,1)$.
\begin{corollary}
	$\lim_{n\to\infty}  C(n,t,3,1)/n^2 =\frac{1}{6(t+1)}$ for any positive integer $t$.
\end{corollary}

 Note that, when $n\equiv 2\pmod{6(t+1)}$, \[
 \left\lfloor\frac{n-2}{3(t+1)}\right\rfloor \cdot \frac{n}{2} + \left\lfloor\frac{n}{2(t+1)}\right\rfloor
 = \frac{(n-2)n}{6(t+1)} + \frac{n-2}{2(t+1)}
 = \frac{n^2+n-6}{6(t+1)} = \left\lfloor\frac{n(n+1)}{6(t+1)}\right\rfloor;
 \]
 when $n\equiv 3t+5 \pmod{6(t+1)}$ is even, for $t>1$ we have
 \[
 \left\lfloor\frac{n-2}{3(t+1)}\right\rfloor \cdot \frac{n}{2} + \left\lfloor\frac{n}{2(t+1)}\right\rfloor
 = \frac{(n-2)n}{6(t+1)} + \frac{n-(t+3)}{2(t+1)}
 = \frac{n^2+n-3(t+3)}{6(t+1)} = \left\lfloor\frac{n(n+1)}{6(t+1)}\right\rfloor,
 \]
 and for $t=1$ it is easy to check that $\lfloor\frac{n-2}{6}\rfloor \cdot \frac{n}{2} + \lfloor\frac{n}{4}\rfloor
 =\frac{n(n+1)}{12}= \lfloor\frac{n(n+1)}{12}\rfloor$.
 Thus for $n\equiv 2 \pmod{3(t+1)}$ with $n$ even we have $\lfloor\frac{n-2}{3(t+1)}\rfloor \cdot \frac{n}{2} + \lfloor\frac{n}{2(t+1)}\rfloor = \lfloor\frac{n(n+1)}{6(t+1)}\rfloor$.
 Observe that for such $n$ we always have $n \equiv 2 \pmod6$ since $n\equiv 2 \pmod{3}$ and $2\mid n$.
 As a consequence, we conclude the following exact value for all $t$.
 \begin{corollary}\label{cor-value-(n,t,3,1)}
 	$C(n,t,3,1) =\lfloor \frac{n(n+1)}{6(t+1)}\rfloor$ for any even integer $n\neq 2$ that satisfies $n\equiv 2 \pmod{3(t+1)}$.
 \end{corollary}

Moreover, by using $4$-frames, $(n,t,4,2)$-LPECC codes can be constructed similarly.
\begin{theorem}\label{thm-(n,t,4,2)-lower}
	For any $n \equiv 3 \pmod{12}$, $C(n,t,4,2)\geq \lfloor\frac{n-3}{4(t+1)}\rfloor \cdot \frac{n}{3} + \lfloor\frac{n}{3(t+1)}\rfloor$.
	Therefore, we have $\liminf_{n\to\infty} C(n,t,4,2)/n^2 \geq \frac{1}{12(t+1)}$.
\end{theorem}

\begin{proof}
	By Lemma~\ref{lem-4-frame}, there exists a $4$-frame $(X,\mathcal{G},\mathcal{B})$ of type $3^{\frac{n}{3}}$. Then it is easy to see that there are exactly $\frac{n}{3}$ holey parallel classes in this $4$-frame. Let $\B_1,\B_2,\dots,\B_{n/3}$ be all holey parallel classes of the $4$-frame $(X,\mathcal{G},\mathcal{B})$. Write $\B_i=\{B_{i,1},B_{i,2},\dots,B_{i,\frac{n-3}{4}}\}$ for each $i\in [n/3]$ and write $\G=\{G_1,\dots,G_{n/3}\}$.	
As in the proof of Theorem~\ref{thm-(n,t,3,1)-lower}, considering  disjoint collections of  $(t+1)$ sets from each of  $\G$ and  $\B_i$ gives a partition with  $\lfloor\frac{n-3}{4(t+1)}\rfloor \cdot \frac{n}{3} + \lfloor\frac{n}{3(t+1)}\rfloor$ parts, which forms an $(n,t,4,2)$-LPECC code.
	%As before, we will let sets in $\G \cup \bigcup_{i\in[n/3]} \B_i$ be blocks of our LPECC code and find a required partition.
	%Let us set \[
%	\PP_j=\{G_{(t+1)j-t},G_{(t+1)j-t+1},\dots,G_{(t+1)j}\},
%	\] for each $j\in [\lfloor\frac{n}{3(t+1)}\rfloor]$, and set \[
%	\Q_{i,k}=\{B_{i,(t+1)k-t},B_{i,(t+1)k-t+1},\dots,B_{i,(t+1)k}\},
%	\] for every $i\in [n/3]$ and $k\in [\lfloor\frac{n-3}{4(t+1)}\rfloor]$.
%	Then $
%	\C := \left( \bigcup_{j\in [\lfloor\frac{n}{3(t+1)}\rfloor]} \PP_j \right)  \cup \left( \bigcup_{i\in [\frac{n}{3}], k\in [\lfloor\frac{n-3}{4(t+1)}\rfloor]} \Q_{i,k}\right)
%	$
%	gives a partition with $\lfloor\frac{n-3}{4(t+1)}\rfloor \cdot \frac{n}{3} + \lfloor\frac{n}{3(t+1)}\rfloor$ parts. It is straightforward to check that $\C$ is an $(n,t,4,2)$-LPECC code.  Thus $C(n,t,4,2)\geq \lfloor\frac{n-3}{4(t+1)}\rfloor \cdot \frac{n}{3} + \lfloor\frac{n}{3(t+1)}\rfloor $.
\end{proof}
    Note that Theorem~\ref{thm-(n,t,w,w-2)-upper} shows that when $w\geq t^2+2t+2$  we have $C(n,t,w,w-2)\leq  \binom{n+1}{2}/\binom{w+t}{2}=\frac{n(n+1)}{(w+t)(w+t-1)}$. On the other hand, our lower bound of $\liminf_{n\to\infty} C(n,t,w,w-2)/n^2$ is of the form $\frac{1}{(t+1)w(w-1)}$ for $w\in \{3,4\}$ and  is tight for $w=3$. Observe that $(t+1)w(w-1)< (w+t)(w+t-1)$  when $w$ is small compared to $t$.
   Based on the above results, it is natural to conjecture the following.
\begin{conjecture}
	If $(t+1)w(w-1)< (w+t)(w+t-1)$, then it holds that
	\[
	 \lim_{n\to\infty} C(n,t,w,w-2)/n^2 = \frac{1}{(t+1)w(w-1)},
	\]
	or even stronger, $C(n,t,w,w-2)=\lfloor \frac{(n+1)n}{(t+1)w(w-1)} \rfloor$ for infinitely many $n$ with fixed $w$ and $t$.
\end{conjecture}

\section{The asymptotic result for $q$-ary $(n,t,w,e)$-LPECC codes}\label{sec-asymptotic}
Given any integer $q\geq 2$, the definition of LPECC codes can be easily extended to $q$-ary codes. One motivation for considering $q$-ary codes is that $q$-ary codes may be useful in constructions of binary codes. Now we state the combinatorial definition of $q$-ary LPECC codes.

Let $Q:=\{0,1,\dots,q-1\}$ be a fixed alphabet of size $q$. Let us write $\bm x:=(x_1,x_2,\dots,x_n)$ for a codeword in $Q^n$.
For any $\bm x \in Q^n$, denote by $\supp(\bm x):= \{i\in [n]: x_i\neq 0\}$ the {\it support} of $\bm x$, and denote by $|\bm x|:=|\supp(\bm x)|$ the {\it weight} of $\bm x$. Let $d(\bm x,\bm y):=|\{i\in [n]: x_i\neq y_i\}|$ denote the {\it Hamming distance} between two codewords $\bm x, \bm y\in Q^n$.
Let $\C\subset Q^n$ be a code with a partition $\PP_1,\PP_2,\dots,\PP_b$. Then $\C=\PP_1\cup \PP_2 \cup \dots\cup \PP_b$ is said to be an $(n,t,w,e)_q$-LPECC code of size $b$ if all of the following properties hold:
\begin{itemize}
	\item[(1)] (Property $\bm{A}_q(t)$) \  For any $T\in \binom{[n]}{t}$ and any $i\in [b]$, there exists $\bm x\in \mathcal{P}_i$ such that $\supp(\bm x) \cap T =\emptyset$;
	\item[(2)] (Property $\bm{B}_q(w)$) \  $|\bm x| \leq w$ for any $\bm x\in \mathcal{C}$;
	\item[(3)] (Property $\bm{C}_q(e)$) \  $d(\bm x,\bm y)\geq 2e+1$, for any $\bm x\in \PP_i, \bm y\in\PP_j$, $i\neq j$, and $i,j\in [b]$.
\end{itemize}
We mention that codes satisfying only Property $\bm A_q(t)$, called $(n, t)_q$-cooling codes, were considered by Chee et al. in \cite[Section~\uppercase\expandafter{\romannumeral5}]{chee2017cooling}.

Let $C_q(n,t,w,e)$ be the maximum size of an $(n,t,w,e)_q$-LPECC code. Clearly, $C(n,t,w,e)=C_2(n,t,w,e)$.
In this section, we give an asymptotic result of LPECC codes in the general $q$-ary setting. We will show that for any $q\geq 2$ and certain fixed parameters $(t,w,e)$,
\[
\lim_{n\rightarrow \infty}\frac{C_q(n,t,w,e)}{(q-1)^{w-e} \binom{n}{w-e}/\binom{w+t}{w-e}}=1.\]
In particular, for $q=2$, this determines the asymptotic behavior of $C(n,t,w,e)$ for a wide range of $(t,w,e)$.

\subsection{An upper bound}
We first show the following upper bound which is asymptotically optimal.
\begin{theorem}\label{thm-(n,t,w,e)-upper-bound}
	Let $q\geq 2$ and let $t,w,e$ be fixed positive integers satisfying $e\leq w-1$ and $2\binom{w}{w-e}  \geq  \binom{w-t-1}{w-e} + \binom{w+t}{w-e}$. If $\mathcal{C}\subset Q^n$ is an $(n,t,w,e)_q$-LPECC code of size $b$, then
	\[
	b \leq \frac{(q-1)^{w-e}\binom{n}{w-e}+\binom{w+t}{w-e}\sum_{k=0}^{w-e-1} (q-1)^{k}\binom{n}{k}/\binom{k+e}{k}}{\binom{w+t}{w-e}}
	\leq (1+o(1)) (q-1)^{w-e}\frac{\binom{n}{w-e}}{\binom{w+t}{w-e}},
	\]
	where $o(1)\rightarrow 0$ as $n\rightarrow \infty$.
\end{theorem}
\begin{remark}\label{rmk-condition}
When $t=1$, it is easy to check that the conditions are satisfied if one assumes $\frac{\sqrt{5}-1}{2}w \leq e\leq w-1$.
For general fixed $t$, assume $d:=w-e\geq 1$ is fixed. Then the polynomial $p_{t,d}(w):=2\binom{w}{d}-\binom{w-t-1}{d}-\binom{w+t}{d}$ is of degree $d-1$ and the leading coefficient of $d!\cdot p_{t,d}(w)$ is
\begin{align*}
 & -2\cdot (1+2+\dots+(d-1)) - (-((t+1)+(t+2)+\dots+(t+d))) - (t+(t-1)+\dots+(t-d+1)) \\
 =& -2\cdot \frac{d(d-1)}{2} + \frac{(2t+d+1)d}{2} - \frac{(2t-d+1)d}{2} =d >0.
\end{align*}
Hence the condition $2\binom{w}{w-e}  \geq  \binom{w-t-1}{w-e} + \binom{w+t}{w-e}$ holds when $w$ is large with $w-e$ and $t$ fixed.
\end{remark}
%Combining Theorem~\ref{thm-(n,t,w,e)-upper-bound} (with $q=2$), Theorem~\ref{thm-packing} and \eqref{eq-near-optimal-packing}, we obtain the following result.
%\begin{corollary}\label{cor-(n,t,w,e)-limit}
%	Let $t,w,e$ be fixed with $e\leq w-1$ and $2\binom{w}{w-e}  \geq  \binom{w-t-1}{w-e} + \binom{w+t}{w-e}$.
%	Then
%	\[
%	\lim_{n\rightarrow \infty}\frac{C(n,t,w,e)}{\binom{n+1}{w-e}/\binom{w+t}{w-e}}=1.\]
%\end{corollary}
To prove Theorem~\ref{thm-(n,t,w,e)-upper-bound}, we further use the idea in the argument of Theorem~\ref{thm-(n,t,w,w-2)-upper} and apply the Johnson-type bound of constant-weight codes (CWCs).  An $(n,d,w)_q$-CWC is a code $\C\subset Q^n$ with each codeword of weight $w$ and any two codewords of Hamming distance at least $d$.
%
%is called a {\it constant-weight code} with weight $w$ if every codeword $\bm x\in \C$ satisfies $|\bm x|=w$. For convenience, write $(n,d,w)_q$-CWC for a constant-weight code $\C\subset Q^n$ with weight $w$ satisfying $d(\bm x,\bm y)\geq d$ for any two different codewords $\bm x,\bm y \in \C$.
Let $A_q(n,d,w)$ denote the maximum cardinality of an $(n,d,w)_q$-CWC. When $d=2e+1$ is odd, the Johnson-type upper bound \cite{ostergaard2002ternary} is
%\begin{lemma}[Johnson-type bound \cite{ostergaard2002ternary}]\label{lemma:Johnson-type}
%$A_q(n,d,w)\leq \lfloor \frac{(q-1)n}{w} A_q(n-1,d,w-1)\rfloor$.
%\end{lemma}
%When $d=2e+1$ is odd, applying Lemma~\ref{lemma:Johnson-type} recursively immediately gives
\begin{equation}\label{equa-cwc-upper}
A_q(n,2e+1,w)\leq \frac{(q-1)^{w-e} \binom{n}{w-e}}{\binom{w}{w-e}}.
\end{equation}
An $(n,2e+1,w)_q$-CWC meeting the upper bound in \eqref{equa-cwc-upper} is called an {\it $(n,2e+1,w)_q$-generalized Steiner system}~\cite{etzion1997optimal}. However, a generalized Steiner system is known to exist only for $q=2$ \cite{Keevash} or  $e\in \{w-1,w-2\}$  \cite{chee2009linear,chee2019decompositions} when $n$ is sufficiently large and satisfies some divisibility conditions. For more on generalized Steiner systems and related problems, one may refer to \cite{etzion2022perfect}.

\begin{proof}[Proof of Theorem~\ref{thm-(n,t,w,e)-upper-bound}]
	Let $\C\subset Q^n$ be an $(n,t,w,e)_q$-LPECC code with a partition $\{\PP_i: i\in [b]\}$.
	From now on, for each $i\in [b]$, let us denote by $\tau(\PP_i)$ the number of $(w-e)$-subsets of $[n]$ covered by some $\supp(\bm x)$ with $\bm x \in\PP_i$, extending its definition for binary codes when $w-e=2$.
	Recall that for any family $\A\subset 2^{X}$ we denote $U(\A):=\cup_{A\in\A} A$.
	Let us define $\supp(\PP_i):= \{\supp(\bm x):\bm x \in\PP_i\}$ and $U(\PP_i) := U(\supp(\PP_i))$ for each $i\in [b]$. Notice that here $U(\PP_i)$ is defined on sets of vectors in $Q^n$.
	\begin{claim}\label{claim-only-w-lower}
		If $\PP_i$ contains only codewords of weight $w$, then $\tau(\PP_i) \geq \binom{w+t}{w-e}$.
	\end{claim}
	\begin{proof}
		It is clear that $|U(\PP_i)|\geq w+t$ by Property $\bm{A}_q(t)$. Furthermore, if $|U(\PP_i)|=w+t$, it follows from Property $\bm{A}_q(t)$ that $\supp(\PP_i)=\binom{U(\PP_i)}{w}$; hence $\tau(\PP_i) = \binom{|U(\PP_i)|}{w-e} = \binom{w+t}{w-e}$.
		Now assume that $|U(\PP_i)|\geq w+t+1$. Take an arbitrary $w$-set $B\in \supp(\PP_i)$ and let $Y:=U(\PP_i)\setminus B$. Then $|Y| \geq t+1$ and we have
		\begin{align*}
		\tau(\PP_i)  \geq & \binom{|B|}{w-e}  + |\{  \{y\}\cup D: y\in Y, D\in\binom{X\setminus \{y\}}{w-e-1},\\
		  & \hspace{4.6cm} \text{ and } \{y\}\cup D \text{ is contained in some block in } \supp(\PP_i) \} | \\				
		 \geq & \binom{w}{w-e}   + \binom{w-1}{w-e-1}+ \binom{w-2}{w-e-1} +\dots+ \binom{w-(t+1)}{w-e-1}  \\
		 =& 2\binom{w}{w-e}  - \binom{w-t-1}{w-e}  \geq  \binom{w+t}{w-e},
		\end{align*}
		where in the last inequality we used the assumption of Theorem~\ref{thm-(n,t,w,e)-upper-bound}. This finishes the proof of Claim~\ref{claim-only-w-lower}.
	\end{proof}
	Let $R_1:=\{i\in [b] : \PP_i \text{ contains only codewords of weight } w \}$ and let $R_2:= [b]\setminus R_1$.
	By Properties $\bm{B}_q(w)$ and $\bm{C}_q(e)$, for every $S\in \binom{[n]}{w-e}$ and every sequence $\bm z \in (Q\setminus\{0\})^{w-e}$, there exists at most one $i\in [b]$ such that $\bm x|_S =\bm z$ for some codeword $\bm x \in\PP_i$, where $\bm x|_S := (x_i:i\in S)$ denotes the restriction of $\bm x$ in $S$.
	It follows from Claim~\ref{claim-only-w-lower} that
	\[
	|R_1|\cdot\binom{w+t}{w-e} \leq \sum_{i\in R_1} \tau(\PP_i) \leq \sum_{i\in [b]} \tau(\PP_i) \leq \binom{n}{w-e}(q-1)^{w-e}.
	\]
	Hence
\begin{equation}\label{consw}
  |R_1|\leq (q-1)^{w-e}\binom{n}{w-e}/\binom{w+t}{w-e}.
\end{equation}
On the other hand, note that for each $i\in R_2$ there exists some codeword $\bm y^i\in \PP_i$ with weight at most $w-1$, and that there exists at most one $i\in R_2$ such that $\PP_i$ contains codewords with weight no more than $e$ by Property $\bm C_q(e)$. Thus, by combining Property $\bm C_q(e)$ and \eqref{equa-cwc-upper}, we have
	 $$
	|R_2|\leq 1+ \sum_{r=e+1}^{w-1} A_q(n,2e+1,r)
	\leq 1+ \sum_{r=e+1}^{w-1} \frac{(q-1)^{r-e} \binom{n}{r-e}}{\binom{r}{r-e}}
	= \sum_{r=e}^{w-1} \frac{(q-1)^{r-e} \binom{n}{r-e}}{\binom{r}{r-e}}
	= \sum_{k=0}^{w-e-1} \frac{(q-1)^{k} \binom{n}{k}}{\binom{k+e}{k}}.
	$$
	 Therefore,
	\[
	b=|R_1|+|R_2|\leq \frac{(q-1)^{w-e}\binom{n}{w-e}}{\binom{w+t}{w-e}} + \sum_{k=0}^{w-e-1} \frac{(q-1)^{k} \binom{n}{k}}{\binom{k+e}{k}}
	\leq (1+o(1)) (q-1)^{w-e}\frac{\binom{n}{w-e}}{\binom{w+t}{w-e}},
	\]
	where $o(1)\to 0$ as $n\to\infty$, completing the proof.
\end{proof}

\subsection{Constructions from constant-weight codes with weight $w+t$}
We use constant-weight codes with weight $w+t$ to construct (constant-weight) $(n,t,w,e)_q$-LPECC codes.
For a codeword $\bm x\in Q^n$, let us define $\pi(\bm x):=\{(i,x_i): i\in [n], x_i\neq 0\}$.

\begin{construction}\label{construction-q-ary}
	Let $\C_0$ be an $(n,2(e+t)+1,w+t)_q$-CWC.
	For any $\bm x\in\C_0$ and $S\subset \supp(\bm x)$, denote by $\bm x^S\in Q^n$ the codeword satisfying $x_i^S=x_i$ if $i\in S$ and $x_i^S=0$ otherwise.
	 For each $\bm x\in\C_0$, let us construct $\PP_{\bm x}$ as follows:
	 \[
	  \PP_{\bm x} := \{ \bm x^S\in Q^n : S\subset \supp(\bm x) \text{ with } |S|=w\}.
	 \]
	 Note that for $\bm x^{S_1} \in \PP_{\bm x}$ and $\bm y^{S_2} \in \PP_{\bm y}$ with different $\bm x,\bm y\in\C_0$,  % and $S_1\in\binom{\supp(\bm x)}{w}, S_2\in\binom{\supp(\bm y)}{w}$,
we have
	 \begin{equation}\label{equa-distance}
	 \begin{aligned}
	 d(\bm x^{S_1},\bm y^{S_2}) &= 2w-|\supp(\bm x^{S_1}) \cap \supp(\bm y^{S_2})| - |\pi(\bm x^{S_1})\cap \pi(\bm y^{S_2})| \\
	 &\geq 2w - |\supp(\bm x) \cap \supp(\bm y)| - |\pi(\bm x)\cap \pi(\bm y)| \\
	 &= 2w- (2(w+t) - d(\bm x,\bm y)) \geq 2e+1.
	 \end{aligned}
	 \end{equation}
	 Hence $\PP_{\bm x}\cap \PP_{\bm y}=\emptyset$ for any two different $\bm x,\bm y\in\C_0$. It is straightforward to check that the code $\C=\bigcup_{\bm x\in\C_0} \PP_{\bm x} \subset Q^n$  with the partition $\{\PP_{\bm x}: \bm x\in\C_0\}$ is an $(n,t,w,e)_q$-LPECC code by the definition and \eqref{equa-distance}.
\end{construction}
 Recall that $A_q(n,d,w)$ denotes the maximum cardinality of an $(n,d,w)_q$-CWC.  By Construction~\ref{construction-q-ary}, we get the following lower bound.
  \begin{theorem}\label{thm-q-ary-construction}
  	$C_q(n,t,w,e) \geq A_q(n,2(e+t)+1,w+t)$.
  \end{theorem}
  Recently, Liu ang Shangguan~\cite{liu2024near} showed for fixed odd distances $d$, there always exist $(n,d,w)_q$-CWCs asymptotically attaining the Johnson-type upper bound in \eqref{equa-cwc-upper}.
  \begin{lemma}[\cite{liu2024near}]\label{lem-near-optimal-CWC}
  	Let $q,w,d$ be fixed positive integers, where $d=2e+1$ is odd.  Then for sufficiently large integer $n$,
  	\[
  	A_q(n,d,w) \geq (1-o(1)) \frac{(q-1)^{w-e} \binom{n}{w-e}}{\binom{w}{w-e}},
  	\]
  	where $o(1)\rightarrow 0$ as $n\rightarrow \infty$.
  \end{lemma}	
  Combining Theorem~\ref{thm-q-ary-construction} and Lemma~\ref{lem-near-optimal-CWC} yields:
  	\begin{corollary}\label{cor-(n,t,w,e)_q-lower}
	Let $q\geq 2$ and let $t,w,e$ be fixed positive integers. Then for sufficiently large integer $n$ we have
	$$C_q(n,t,w,e) \geq (1-o(1)) (q-1)^{w-e}\frac{\binom{n}{w-e}}{\binom{w+t}{w-e}},$$ where $o(1)\rightarrow 0$ as $n\rightarrow \infty$.
\end{corollary}
 By combining Theorem~\ref{thm-(n,t,w,e)-upper-bound}, Remark~\ref{rmk-condition}, and Corollary~\ref{cor-(n,t,w,e)_q-lower}, the asymptotic behavior of $C_q(n,t,w,e)$ is determined when $2\binom{w}{w-e}  \geq  \binom{w-t-1}{w-e} + \binom{w+t}{w-e}$.

\begin{corollary}\label{cor-limit-q-ary-lpecc}
	Let $q\geq 2$, and let $t,w,e$ be fixed positive integers with $e\leq w-1$ and $2\binom{w}{w-e}  \geq  \binom{w-t-1}{w-e} + \binom{w+t}{w-e}$. Then,
\[
\lim_{n\rightarrow \infty}\frac{C_q(n,t,w,e)}{(q-1)^{w-e} \binom{n}{w-e}/\binom{w+t}{w-e}}=1.\]
In particular, when $q=2$, $\lim_{n\rightarrow \infty}\frac{C(n,t,w,e)}{ \binom{n}{w-e}/\binom{w+t}{w-e}}=1$ when $w-e$ and $t$ are fixed and $w$ is large enough.
\end{corollary}

\begin{remark}\label{rem}
	Note that here we all consider the asymptotic behavior of $C_q(n,t,w,e)$ for sufficiently large $n$. When $n,w,e$ are fixed integers, Chee and Ling~\cite{chee2007constructions} showed that for any real number $\epsilon>0$ and sufficiently large integer $q$,
	\[
		A_q(n,d,w) \geq (1-o(1)) \frac{(q-1)^{w-e-\epsilon} \binom{n}{w-e}}{\binom{w}{w-e}},
	\]
	where $o(1)\rightarrow 0$ as $q\rightarrow \infty$. Observe that Theorem~\ref{thm-(n,t,w,e)-upper-bound} also gives a near-optimal upper bound with respect to $q$. Therefore, if $n,t,w,e$ are fixed positive integers with $e\leq w-1$ and $2\binom{w}{w-e}  \geq  \binom{w-t-1}{w-e} + \binom{w+t}{w-e}$, then for any real number $\epsilon>0$ and sufficiently large integer $q$ we have
	\[
	  (1-o(1)) \frac{(q-1)^{w-e-\epsilon} \binom{n}{w-e}}{\binom{w+t}{w-e}} \leq C_q(n,t,w,e) \leq (1+o(1)) \frac{(q-1)^{w-e} \binom{n}{w-e}}{\binom{w+t}{w-e}},
	\]
	where $o(1)\rightarrow 0$ as $q\rightarrow \infty$.
\end{remark}

\subsection{Constant-power error-correcting cooling (CPECC) codes}\label{sec-cpecc}
In \cite[Section \uppercase\expandafter{\romannumeral4}]{chee2020low}, Chee et al. also considered constant-power error-correcting cooling (CPECC) codes, which require that each block has a constant size. Here we discuss general $q$-ary CPECC codes.

Let $\C\subset Q^n$ be a $q$-ary code with a partition $\PP_1,\PP_2,\dots,\PP_b$. Then $\C=\PP_1\cup \PP_2 \cup \dots\cup \PP_b$ is said to be an $(n,t,w,e)_q$-CPECC code of size $b$ if it is an $(n,t,w,e)_q$-LPECC code such that
\begin{itemize}
	\item (Property $\bm{B}'_q(w)$) \  $|\bm x| = w$ for any $\bm x\in \mathcal{C}$.
\end{itemize}
Let $C'_q(n,t,w,e)$ be the maximum size of an $(n,t,w,e)_q$-CPECC code. Indeed, (\ref{consw}) and Construction~\ref{construction-q-ary} already give an upper bound and a lower bound of $C'_q(n,t,w,e)$ respectively. %So we can conclude the following result. %note that we indeed constructed an $(n,t,w,e)_q$-CPECC code  in; hence we have the same natural lower bound.

%\begin{proposition}\label{prop-cpecc-lower}
%	 $C'_q(n,t,w,e)\geq  A_q(n,2(e+t)+1,w+t)$.
%\end{proposition}
%
%By (\ref{consw}), we get the following upper bound.

 % In 2014, Keevash~\cite{Keevash} proved that
% $
% A_2(n,2e+1,w)=\binom{n}{w-e}/\binom{w}{w-e},
% $
% provided that $n$ is large enough and satisfies certain necessary divisibility conditions, i.e., $\binom{w-i}{w-e-i} \mid \binom{n-i}{w-e-i}$ for every $0\leq i \leq w-e-1$. More generally, an $(n,2e+1,w)_q$-CWC meeting the upper bound in \eqref{equa-cwc-upper} is called an {\it $(n,2e+1,w)_q$-generalized Steiner system}~\cite{etzion1997optimal}. That means, $A_q(n,2e+1,w)=(q-1)^{w-e}\binom{n}{w-e}/\binom{w}{w-e}$ if and only if an $(n,2e+1,w)_q$-generalized Steiner system exists. However, for generalized Steiner systems, no general existence result as Keevash's exact result is known. When $e\in \{w-1,w-2\}$, the existence of $(n,2e+1,w)_q$-generalized Steiner systems (with $n$ sufficiently large and satisfying some divisibility conditions) has been proved in \cite{chee2009linear,chee2019decompositions}. For more on generalized Steiner systems and related problems, one may refer to \cite{etzion2022perfect}.

 \begin{theorem}\label{thm-cpecc-upper}
 	Let $q\geq 2$ and let $t,w,e$ be positive integers with $e\leq w-1$. Then
 	\[
 A_q(n,2(e+t)+1,w+t)	\leq C'_q(n,t,w,e)\leq  \frac{(q-1)^{w-e} \binom{n}{w-e}}{\binom{w+t}{w-e}},
 	\]where the upper bound requires $2\binom{w}{w-e}  \geq  \binom{w-t-1}{w-e} + \binom{w+t}{w-e}$.
 \end{theorem}

 So Corollary \ref{cor-limit-q-ary-lpecc} and Remark \ref{rem} also hold for $C'_q(n,t,w,e)$. Further, we have the following exact results based on the existence of generalized Steiner systems \cite{chee2009linear,chee2019decompositions,Keevash}.
 \begin{corollary}\label{cor-value-cpecc}
Let $t,w,e$ be positive integers satisfying $e\leq w-1$ and $2\binom{w}{w-e}  \geq  \binom{w-t-1}{w-e} + \binom{w+t}{w-e}$.
 \begin{itemize}\item[(a)] When $q=2$, $C'_2(n,t,w,e) =  \binom{n}{w-e} / \binom{w+t}{w-e}$, provided that $n$ is large enough and satisfies  $\binom{w+t-i}{w-e-i} \mid \binom{n-i}{w-e-i}$ for any $0\leq i \leq w-e-1$.
 \item[(b)] When $e\in \{w-1,w-2\}$, $
	C'_q(n,t,w,e) = \frac{(q-1)^{w-e} \binom{n}{w-e}}{\binom{w+t}{w-e}}
	$ for infinitely many $n$.
	%whenever $n$ is sufficiently large and satisfies some divisibility conditions.
 \end{itemize}
 \end{corollary}

\section{Concluding remarks}\label{conclusion}
 In this paper, we mainly investigated the maximum possible sizes of LPECC codes with large Hamming distances. For $(n,t,w,w-2)$-LPECC codes with  $w$ large compared to $t$, we proved a tight upper bound and characterized the unique extremal structure. For $(n,t,w,w-2)$-LPECC codes with small $w$, we determined the exact value of $C(n,t,3,1)$ when $n$ satisfies some divisibility conditions and gave a lower bound on $C(n,t,4,2)$. We believe that this lower bound on $C(n,t,4,2)$ is (asymptotically) sharp  for large $t$. Moreover, we determined the asymptotic behavior of $C(n,t,w,e)$ for a wide range of parameters by studying  $q$-ary LPECC codes. Finally, we also discussed the exact values of the maximum sizes of general $q$-ary constant-weight LPECC codes, i.e., CPECC codes. It would be of interest to give more upper bounds on the size of LPECC codes for general parameters, and provide more (asymptotically) tight constructions.

%\bibliographystyle{abbrv}
%\bibliography{LPECC_codes}

\end{document}